\documentclass{llncs}
\usepackage{amsfonts}
\usepackage{xcolor}
\usepackage{graphicx}
\usepackage{amssymb}
\usepackage{amsmath}
\usepackage{stmaryrd}

\usepackage{algorithm}
\usepackage[indLines=false]{algpseudocodex}

\usepackage{hyperref}
\usepackage{cleveref}
\usepackage{xspace}
\usepackage{marginnote}
\usepackage{multirow}
\usepackage{rotating}
\usepackage{colortbl}
\usepackage{thm-restate}
\usepackage{enumitem}
\usepackage{orcidlink} 
\usepackage[misc,geometry]{ifsym}
\usepackage{caption}
\usepackage{subcaption}

\usepackage{array}
\usepackage{booktabs}
\usepackage{soul}
\usepackage{hhline}

\usepackage{adjustbox}
\usetikzlibrary{tikzmark}

\crefformat{line}{#2l. #1#3}
\crefrangeformat{line}{ll. #3#1#4--#5#2#6}
\crefmultiformat{line}{ll. #2#1#3}{ and #2#1#3}{, #2#1#3}{ and #2#1#3}

\usepackage{pgfplots}
\usepackage{tikz}

\usetikzlibrary{arrows,automata,fit,positioning,shadows,shapes}
\usetikzlibrary{decorations,snakes,calc} 
\usetikzlibrary{patterns, patterns.meta}
\tikzstyle{location}=[rectangle, rounded corners, minimum size=12pt, draw=black, fill=blue!10, inner sep=2pt]
\tikzstyle{pta}=[auto, ->, >=stealth']
\tikzstyle{PZG}=[auto, ->, >=stealth']
\tikzstyle{mergingFigure} = [>=stealth', node distance=1.8cm, yscale=.6]
\tikzstyle{location10}=[location, minimum size=10pt]
\tikzstyle{invariant}=[draw=black, dotted, inner sep=1pt, node distance=0] 
\tikzstyle{final}=[fill=green!70,double]
\tikzstyle{urgent}=[dotted, draw=red, very thick, fill=yellow]
\tikzstyle{bad}=[fill=red]

\tikzstyle{pzgstate} = [
	rectangle split,
	rectangle split horizontal,
	rectangle split parts=2,
	draw,
	rectangle split part align={center,left},
	rounded corners,
	minimum height=0.5 cm,
	minimum width=1.5 cm,
	thick
	]
\tikzstyle{fillred} = [ fill=red!20 ]
\tikzstyle{fillblue} = [ fill=blue!20 ]
\tikzstyle{fillyellow} = [ fill=yellow!20 ]
\tikzstyle{line} = [ draw,-latex',thick ]
\tikzstyle{highlightarrow} = [
	preaction={
		draw,
		line width=0.5mm
	}
]

\tikzstyle{urgent}=[fill=yellow, thick, dotted] 
\tikzstyle{private}=[fill=red!50,thick]

\pgfdeclarepatternformonly{north west lines wide}
{\pgfqpoint{-1pt}{-1pt}}
{\pgfqpoint{10pt}{10pt}}
{\pgfqpoint{9pt}{9pt}}
{
	\pgfsetlinewidth{3pt}
	\pgfpathmoveto{\pgfqpoint{0pt}{9.1pt}}
	\pgfpathlineto{\pgfqpoint{9.1pt}{0pt}}
	\pgfusepath{stroke}
}

\tikzset{onslide/.code args={<#1>#2}{%
		\only<#1>{\pgfkeysalso{#2}} 
}}

\definecolor{coloract}{rgb}{0.50, 0.70, 0.30}
\definecolor{colorclock}{rgb}{0.4, 0.4, 1}
\definecolor{colordisc}{rgb}{1, 0, 1}
\definecolor{colorloc}{rgb}{0.4, 0.4, 0.65}

\definecolor{colorparam}{rgb}{.66, 0.4, 0.0}
\definecolor{colorstate}{rgb}{1, 0.4, 0.0}

\definecolor{loccolor1}{rgb}{1, 0.6, 0.45}
\definecolor{loccolor2}{rgb}{0.45, 1, 0.45}
\definecolor{loccolor3}{rgb}{0.8, 0.8, 1}
\definecolor{loccolor4}{rgb}{1, 0.45, 1}
\definecolor{loccolor5}{rgb}{1, 1, 0.45}
\definecolor{loccolor6}{rgb}{0.45, 1, 1}
\definecolor{loccolor7}{rgb}{0.9, 0.6, 0.2}
\definecolor{loccolor8}{rgb}{0.7, 0.4, 1}
\definecolor{loccolor9}{rgb}{0.5, 1, 0.75}
\definecolor{loccolor10}{rgb}{0.8, 0.7, 0.6}
\definecolor{loccolor11}{rgb}{0.6, 0.7, 0.8}
\definecolor{loccolor12}{rgb}{0.2, 0.5, 0.9}
\definecolor{loccolor13}{rgb}{0.5, 0.9, 0.2}
\definecolor{loccolor14}{rgb}{0.9, 0.2, 0.5}
\definecolor{loccolor15}{rgb}{0.7, 0.7, 0.7}
\definecolor{loccolor16}{rgb}{0.8, 0.8, 0.5}

\newcommand{\styleact}[1]{\ensuremath{\textcolor{coloract}{\mathrm{#1}}}}
\newcommand{\styleclock}[1]{\ensuremath{\textcolor{colorclock}{\mathrm{#1}}}}

\newcommand{\styleloc}[1]{\ensuremath{\textcolor{colorloc}{\mathrm{#1}}}}
\newcommand{\styleparam}[1]{\ensuremath{\textcolor{colorparam}{\mathrm{#1}}}}

\usepackage{xcolor}
\definecolor{darkgreen}{rgb}{0.0, 0.4, 0.08}
\definecolor{lighterblack}{rgb}{.4, .4, .4}
\definecolor{colorparam}{rgb}{1, 0.6, 0.0}
\definecolor{mygreen}{rgb}{0,0.6,0}
\definecolor{mygray}{rgb}{0.5,0.5,0.5}
\definecolor{mymauve}{rgb}{0.58,0,0.82}

\definecolor{gris}{rgb}{0.6,0.6,0.6}
\definecolor{grisfonce}{rgb}{0.2,0.2,0.2}
\definecolor{turquoise}{rgb}{0, 1, 1}
\definecolor{vertfonce}{rgb}{0,0.85,0}
\definecolor{violet}{rgb}{0.8,0,0.8}
\definecolor{grispale}{rgb}{0.9, 0.9, 0.9}
\definecolor{cpale1}{rgb}{1, 0.3, 0.3}
\definecolor{cpale2}{rgb}{0.3, 1, 0.3}
\definecolor{cpale3}{rgb}{0.3, 0.3, 1}
\definecolor{cpale4}{rgb}{1, 0.3, 1}
\definecolor{cpale5}{rgb}{1, 1, 0.3}
\definecolor{cpale6}{rgb}{0.3, 1, 1}
\definecolor{cpale7}{rgb}{0.9, 0.6, 0.2}
\definecolor{cpale8}{rgb}{0.7, 0.4, 1}
\definecolor{cpale9}{rgb}{0.5, 1, 0.75}
\definecolor{cpale10}{rgb}{0.8, 0.7, 0.6}
\definecolor{cpale11}{rgb}{0.6, 0.7, 0.8}
\definecolor{cpale12}{rgb}{0.2, 0.5, 0.9}
\definecolor{cpale13}{rgb}{0.5, 0.9, 0.2}
\definecolor{cpale14}{rgb}{0.9, 0.2, 0.5}
\definecolor{cpale15}{rgb}{0.7, 0.7, 0.7}
\definecolor{cpale16}{rgb}{0.8, 0.8, 0.5}
\definecolor{bleuciel}{rgb}{0.90,0.95,1}
\definecolor{cv1}{rgb}{1, 0, 0}
\definecolor{cv2}{rgb}{0, 1, 0}
\definecolor{cv3}{rgb}{0, 0, 1}
\definecolor{cv4}{rgb}{1, 1, 0}
\definecolor{cv5}{rgb}{1, 0, 1}
\definecolor{cv6}{rgb}{0, 1, 1}
\definecolor{cv7}{rgb}{0.8, 0.6, 0.4}
\definecolor{cv8}{rgb}{0.5, 0.5, 1}
\definecolor{cv9}{rgb}{0.55, 0.75, 0.35}
\definecolor{cv10}{rgb}{1, 0.6, 0.1}
\definecolor{cv11}{rgb}{0.6, 0.7, 0.8}
\definecolor{cv12}{rgb}{0.2, 0.5, 0.9}
\definecolor{cv13}{rgb}{0.5, 0.9, 0.2}
\definecolor{cv14}{rgb}{1, 0.3, 0.5}
\definecolor{cv15}{rgb}{0.7, 0.7, 0.7}
\definecolor{cv16}{rgb}{0.8, 0.8, 0.5}
\definecolor{cvorange}{rgb}{1,.8,0.5}
\definecolor{colortask}{rgb}{0.5, 0.2, 0.9} 

\newcommand{\marginX}{\marginnote{\huge{\quad\textbf{!}\quad}}}

\ifdefined\VersionWithComments
  \newcommand{\mbdj}[1]{\textcolor{purple}{\marginX{}[\textbf{Mikael}: #1]}}
  \newcommand{\bfi}[1]{\textcolor{orange}{\marginX{}[\textbf{Baptiste}: #1
  ]}}
  \newcommand{\lp}[1]{\textcolor{green!70!black}{\marginX{}[\textbf{Laure}: #1 ]}}
  \newcommand{\jvdp}[1]{\textcolor{blue!40}{\marginX{}[\textbf{Jaco}:
  #1 ]}}

\else
  \newcommand{\mbdj}[1]{}
  \newcommand{\bfi}[1]{}
  \newcommand{\lp}[1]{}
  \newcommand{\jvdp}[1]{}
\fi

\newcommand{\ie}{i.e.\xspace}
\newcommand{\wrt}{w.r.t.\xspace}

\newcommand{\game}{\ensuremath{G} }

\newcommand{\LocSet}{\ensuremath{L}}
\newcommand{\loc}{\ensuremath{\ell} }

\newcommand{\ClockSet}{\ensuremath{X}}

\newcommand{\ParamSet}{\ensuremath{P}}

\newcommand{\TransSet}{\ensuremath{T} }
\newcommand{\trans}{\ensuremath{t} }

\newcommand{\temptrans}{\to^{\delta}}
\newcommand{\disctrans}{\to^{t}}

\newcommand{\src}[1]{\ensuremath{\mathit{src}(#1)}}
\newcommand{\dec}[1]{\ensuremath{\mathit{dec}(#1)}}

\newcommand{\guard}{\ensuremath{g} }
\newcommand{\guardFunction}[1]{\ensuremath{\mathit{guard}(#1)}}

\newcommand{\SubsetSet}[1]{\ensuremath{\mathcal{P}(#1)} }

\newcommand{\LabelSet}{\ensuremath{Act} }

\newcommand{\Inv}{\ensuremath{\mathit{Inv}}}

\newcommand{\val}{\ensuremath{v} }
\newcommand{\ClockValSet}[1]{ \ensuremath{\mathbb{R}_{\geq 0}^{#1}}}
\newcommand{\ParamValSet}[1]{ \ensuremath{\mathbb{Q}_{\geq 0}^{#1}}}
\newcommand{\ValSet}{\ensuremath{V}}

\newcommand{\state}{\ensuremath{s} }
\newcommand{\StateSpace}{\ensuremath{\mathbb{S}}}
\newcommand{\StateSet}{\ensuremath{S}}

\newcommand{\ParamLinearTerm}{\ensuremath{plt} }

\newcommand{\ZoneFormula}{\phi}

\newcommand{\delay}{\ensuremath{\delta} }

\newcommand{\RunSet}{\ensuremath{\mathcal{R}} }
\newcommand{\run}{\ensuremath{r} }

\newcommand{\HistSet}{\ensuremath{\mathcal{H}} }
\newcommand{\hist}{\ensuremath{h} }

\newcommand{\LastState}{ \ensuremath{ls} }

\newcommand{\TargetSet}{\ensuremath{{R}} }

\newcommand{\zone}{\ensuremath{Z} }
\newcommand{\ZoneSet}{\ensuremath{\mathcal{Z}} }

\newcommand{\strat}{\sigma}

\newcommand{\SymbState}{\ensuremath{\xi} }

\newcommand{\TempPred}[1]{ \ensuremath{ #1^{\swarrow} } }
\newcommand{\TempSucc}[1]{ \ensuremath{ #1^{\nearrow} } }

\newcommand{\ProjectParam}[1] { \ensuremath{#1{\downarrow_P}}}

\newcommand{\imitator}{\textsc{Imitator}\xspace}
\newcommand{\Uppaal}{\textsc{Uppaal}\xspace}

\newcommand{\Win}{\ensuremath{\mathit{Win}}}

\newcommand{\Depends}{\ensuremath{\mathit{Depends}}}
\newcommand{\Pred}{\ensuremath{\mathit{Pred}}}
\newcommand{\SafePred}{\ensuremath{\mathit{SafePred}}}
\newcommand{\NewWin}{\ensuremath{\mathit{NewWin}}}

\newcommand{\WinningParam}{\ensuremath{\mathit{WinningParam}}}
\newcommand{\Uncontrollable}{\ensuremath{\mathit{UnCtrl}}}
\newcommand{\Controllable}{\ensuremath{\mathit{Ctrl}}}
\newcommand{\UTempSplit}{\ensuremath{\mathit{UTempSplit}}}

\newcommand{\Explored}{\ensuremath{\mathit{Explored}}}
\newcommand{\WaitingExplore}{\ensuremath{\mathit{WaitingExplore}}}
\newcommand{\WaitingUpdate}{\ensuremath{\mathit{WaitingUpdate}}}
\definecolor{ColorLosingProp}{HTML}{1B9E77}
\definecolor{ColorCumuPrune}{HTML}{D95F02}
\definecolor{ColorCovPrune}{HTML}{7570B3}
\definecolor{ColorInc}{HTML}{E7298A}
\newcommand{\ColorBox}[2]{
  \BeginBox[draw=#1, fill=#1!20!white]
  #2
  \EndBox
}
\newcommand{\ColorString}[2]{
  \BoxedString[draw=#1, fill=#1!20!white]{#2}
}

\crefname{definition}{Def.}{Defs.}
\crefname{theorem}{Thm.}{Thms.}
\crefname{lemma}{Lem.}{Lemmas}
\crefname{line}{Line}{Lines}
\crefname{algorithm}{Alg.}{Algs.}
\crefname{figure}{Fig.}{Figs.}
\crefname{appendix}{Appendix}{Appendices}
\crefname{section}{Sec.}{Sections}
\crefname{table}{Table}{Tables}

\newcommand{\decision}{\ensuremath{d}}

\newcommand{\ControllerLoc}{\ensuremath{q}}

\newcommand{\UpTempCons}{\ensuremath{\ZoneFormula}}
\newcommand{\LoTempCons}{\ensuremath{\ZoneFormula}}

\newcommand{\UpTempConsSet}{\ensuremath{U_{\mathit{Temp}}}}
\newcommand{\LoTempConsSet}{\ensuremath{L_{\mathit{Temp}}}}

\newcommand{\UpTempBound}[2]{\partial #1^{\uparrow}_{#2}}
\newcommand{\LoTempBound}[2]{\partial #1^{\downarrow}_{#2}}

\newcommand{\UpTempGlobBound}[1]{\partial #1^{\uparrow}}
\newcommand{\LoTempGlobBound}[1]{\partial #1^{\downarrow}}

\newcommand{\UpTempClos}[1]{\overline{#1^{\uparrow}}}
\newcommand{\LoTempClos}[1]{\overline{#1^{\downarrow}}}

\newcommand{\LoTempBoundEpsilon}[2]{\partial #1_{#2}^{\downarrow \epsilon}}

\newcommand{\DelaySet}{\mathbb{R}^{\infty}_{\geq 0}}

\newcommand{\PartList}{L}
\newcommand{\PartListAlg}{\ensuremath{\mathit{InstrList}}}

\newcommand{\Wait}{\ensuremath{\mathit{Wait}}}

\newcommand{\ForcedMoves}{\ensuremath{\mathit{ForcedMoves}}}
\newcommand{\WinningMove}{\ensuremath{\mathit{WinningMove}}}
\newcommand{\InvBound}{\ensuremath{\mathit{InvBound}}}

\newcommand{\instr}{\ensuremath{I}}
\newcommand{\ind}{\ensuremath{k}}

\newcommand{\cont}{\mathcal{C}}
\newcommand{\ParallelComp}[2]{( #1 \parallel #2 )}

\RequirePackage{graphicx}
\usepackage[firstpageonly=true,angle=0,vpos=.147\paperheight,hpos=.393\linewidth,vanchor=t,hanchor=l]{draftwatermark}
\SetWatermarkText{
\raisebox{-0cm}
{\begin{minipage}{\linewidth}\centering\includegraphics[width=12.5mm]{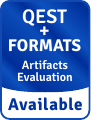}\hfill\includegraphics[width=12.5mm]{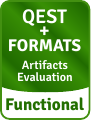}\end{minipage}}%
}

\title{Controller Synthesis for \\ Parametric Timed Games
     \thanks{This work was partially supported by
         CNRS international PhD programme,
         the CNRS International Research Network CLoVe
         and Innovationsfonden Danmark's DIREC project SIoT (Secure Internet of Things).
         }
    \thanks{This is the full version of the paper under the same title accepted to QEST+FORMATS 2025.}
     }

\author{
 Mikael Bisgaard Dahlsen-Jensen \inst{1}\orcidlink{0000-0003-0641-7635}
 \and
 Baptiste Fievet \inst{2}\orcidlink{0000-0002-4925-1105}
 \and \\
 Laure Petrucci \Letter\inst{2}\orcidlink{0000-0003-3154-5268}
 \and
 Jaco van de Pol \inst{1}\orcidlink{0000-0003-4305-0625}
 }

\institute{
 Aarhus University, Aarhus, Denmark
 \\\email{\{mikael,jaco\}@cs.au.dk}
 \and
 CNRS, Université Sorbonne Paris Nord, LIPN, F-93430 Villetaneuse, France
 \\ \email{\{Baptiste.Fievet, Laure.Petrucci\}@lipn.univ-paris13.fr}
 }

 \authorrunning{Dahlsen-Jensen, Fievet, Petrucci, van de Pol}

\titlerunning{Controller Synthesis for Parametric Timed Games}

\begin{document}

\maketitle

\setlength{\abovedisplayskip}{0pt}
\setlength{\belowdisplayskip}{0pt}
\begin{abstract}
    We present a (semi)-algorithm to compute winning strategies for parametric timed games.
    Previous algorithms only synthesized constraints on the clock parameters for which the game is winning.
    A new definition of (winning) strategies is proposed, and ways to compute them.
    A transformation of these strategies to (parametric) timed automata allows for
    building a controller enforcing them.
    The feasibility of the method is demonstrated by an implementation 
    and experiments
    for the Production Cell case study.
\end{abstract}

\section{Introduction}
\label{sec:intro}

Timed Games (TG)~\cite{Classic-TG} extend Timed Automata (TA)~\cite{Intro-TA} by distinguishing controllable 
and uncontrollable transitions and introducing a reachability goal. The game
is won by the controller, if he can play controllable actions, such that no matter
which uncontrollable actions are taken by the environment, the goal is reached.
An on-the-fly algorithm was introduced in \cite{OTF-TG},
to decide whether a timed game is won by the controller. This algorithm forms the basis of \Uppaal{} Tiga \cite{DBLP:conf/cav/BehrmannCDFLL07}.
Parametric Timed Automata (PTA)~\cite{DBLP:conf/stoc/AlurHV93} enable the study of an infinite family of TA by replacing concrete constraints by parameters. Given a desired property, the problem is to find all parameter valuations that satisfies it. Unlike TA, most problems are undecidable for PTA~\cite{DBLP:journals/sttt/Andre19}, although several fragments are known to be decidable, e.g., L/U PTA~\cite{DBLP:journals/jlp/HuneRSV02} and PTA with only one clock~\cite{BUNDALA2017272}.

Parametric Timed Games (PTG) combine the ideas of PTA and TG and were introduced along
with a semi-algorithm for parameter synthesis for reachability objectives~\cite{DBLP:conf/wodes/JovanovicFLR12,DBLP:journals/ijcon/JovanovicLR19}, extended and implemented in
\cite{dahlsenjensen2024ontheflyalgorithmreachabilityparametric}. The semi-algorithm
in that work enumerates all constraints on clock parameters, for which the game is
won.

\medskip
We extend the previous work in several directions. Our main goal is to generate a \emph{concrete strategy}, that tells the controller how to win the game. Thus, we automatically synthesize a controller that is correct by construction.
This strategy is represented itself by a parametric timed automaton, which can be put in parallel to the original system.
In the synchronous product of the controller and the system, every run will eventually reach the goal location.

\Uppaal{} Tiga~\cite{DBLP:conf/cav/BehrmannCDFLL07} allows generating winning strategies for TG and~\cite{DBLP:conf/formats/DavidFLZ14} proposes a translation from these to controller timed automata. However, these strategies have the ``implicit semantics" that their instructions must be taken \emph{as soon as possible}, causing ill-defined behavior in some cases.
We improve the situation by a new type of strategies, whose instructions explicitly define the interval in which
a transition must be taken, and extend this to PTG. This new notion of strategy allows
for expressing the controller as a
Parametric Timed Automaton. 

We also reconciled the algorithm from \cite{OTF-TG} with the actual implementation in \cite{DBLP:conf/cav/BehrmannCDFLL07}.
In \cite{OTF-TG} (and also in its extension \cite{dahlsenjensen2024ontheflyalgorithmreachabilityparametric}), a strategy can only rely on controllable actions, while \Uppaal{} Tiga takes the more realistic approach, 
in which the strategy can
also consider uncontrollable actions, as long as they cannot be postponed forever.
In this paper, we develop the theory of these \emph{forced uncontrollable transitions}, and adapt the algorithms
accordingly.

\subsection{Motivating Example} \label{motiv}

Let us first recall how~\cite{OTF-TG}, its parametric extension~\cite{DBLP:conf/wodes/JovanovicFLR12,DBLP:journals/ijcon/JovanovicLR19} as well as the tool \Uppaal{} Tiga itself defines strategies: they are memoryless, \ie{} provide
the next action based solely on the current state. This action can be either a discrete
transition or a \emph{Wait} instruction. In case of a \emph{Wait}, the controller
must wait until a state is reached where the strategy provides a discrete
transition, which should then be taken \emph{as soon as possible}.
However, the ``first" future state where the strategy outputs a
discrete transition may not exist.

\begin{example}
Consider the example in~\cref{fig:strat_demo:1}.
It is a Timed Game with a clock $\styleclock{x}$ and a reachability
objective with $\styleloc{Win}$ as the only target location. Controllable (resp. uncontrollable) transitions are drawn by solid (resp. dashed) arrows, have action labels, optional guards and clock resets.

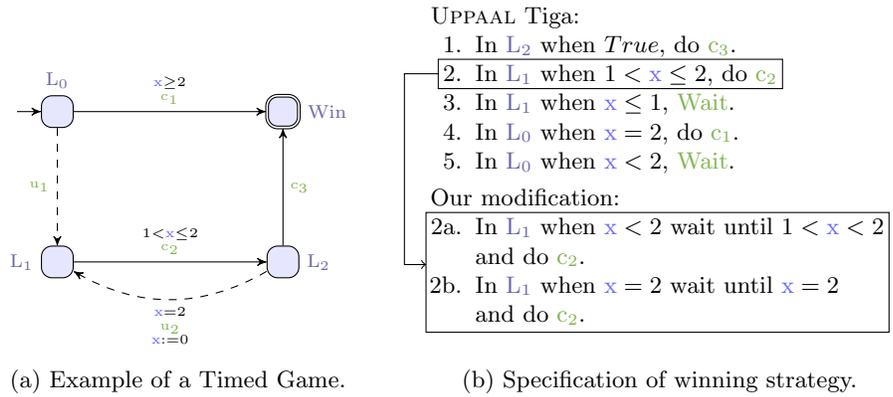
\begin{figure}[b]
    \begin{subfigure}[b]{0.45\textwidth} \centering\scriptsize
        \begin{tikzpicture}[>=stealth', node distance=3cm,initial text=]
            \node[location,initial left,label={above:\styleloc{L_0}}] (l0) {};
            \node[location,label={left:\styleloc{L_1}},below of=l0,
                node distance=2cm] (l1) {};
            \node[location,label={right:\styleloc{L_2}},right of=l1] (l2) {};
            \node[location,accepting,label={right:\styleloc{Win}},above of=l2, node distance=2cm] 
            (win) {};
            
            \draw[->,dashed] (l0) -- node [left] 
            {$\substack{\styleact{u_1}}$} (l1);
            \draw[->,dashed] (l2) edge [bend left] node [below]
                {$\substack{\styleclock{x} = 2 \\
                    \styleact{u_2} \\
                    \styleclock{x}:=0}$} (l1);
            \draw[->] (l0) -- node [above]
                {$\substack{\styleclock{x} \geq 2 \\
                    \styleact{c_1}}$} (win);
            \draw[->] (l1) -- node [above] 
                {$\substack{1 < \styleclock{x} \leq 2 \\
                    \styleact{c_2}}$} (l2);
            \draw[->] (l2) -- node [right]
                {$\substack{\styleact{c_3}}$} (win);
        \end{tikzpicture}
        \caption{Example of a Timed Game.\label{fig:strat_demo:1}}
    \end{subfigure}
    \hfill
    \begin{subfigure}[b]{0.5\textwidth}
        \Uppaal{} Tiga:
        \begin{enumerate}[nolistsep]
            \item In $\styleloc{L_2}$ when $True$, do $\styleact{c_3}$.
            \item[\tikzmark{A}2.] In $\styleloc{L_1}$ when $1 < \styleclock{x} \leq 2$, do $\styleact{c_2}$\tikzmark{B}
            \setcounter{enumi}{2}
            \item In $\styleloc{L_1}$ when $\styleclock{x} \leq 1$, $\styleact{Wait}$.
            \item In $\styleloc{L_0}$ when $\styleclock{x} = 2$, do $\styleact{c_1}$.
            \item In $\styleloc{L_0}$ when $\styleclock{x} < 2$, $\styleact{Wait}$.
        \end{enumerate}

        \smallskip

        Our modification:
        \begin{enumerate}[nolistsep]
            \item[\tikzmark{C}2a.] In $\styleloc{L_1}$ when $\styleclock{x} < 2$ wait until $1 < \styleclock{x} < 2$\tikzmark{D}\\ and do $\styleact{c_2}$.
            \item[2b.] In $\styleloc{L_1}$ when $\styleclock{x} = 2$ wait until $\styleclock{x} = 2$\\ and do $\styleact{c_2}$.
        \end{enumerate}
        \begin{tikzpicture}[overlay,remember picture]
            \draw[->] ([xshift=-.2em,yshift=1mm]pic cs:A) --++ (-.45,0) |- ([xshift=-.2em,yshift=-4mm]pic cs:C);
            \draw ([xshift=-.2em,yshift=.9em]pic cs:A) rectangle ([xshift=.2em,yshift=-.3em]pic cs:B);
            \draw ([xshift=-.2em,yshift=.9em]pic cs:C) rectangle ([xshift=.2em,yshift=-4em]pic cs:D);
        \end{tikzpicture}
        \caption{Specification of winning strategy.\label{fig:strat_demo:2}}
    \end{subfigure}
    \caption{Timed Game, strategy from \Uppaal{} Tiga, and a modified fix.}
\end{figure}

The environment can force a transition from $\styleloc{L_0}$ to $\styleloc{L_1}$ for any value $0\leq \styleclock{x}\leq 2$.
When in $\styleloc{L_1}$ there are three cases: 
If $1 < \styleclock{x} < 2$, we can immediately reach
$\styleloc{Win}$ by taking $\styleact{c_2}$ and $\styleact{c_3}$. If $\styleloc{x} \leq 1$, we can safely wait
to reach the interval $1 < \styleclock{x} < 2$, and apply the previous case.
If $\styleclock{x} = 2$, we should still take $\styleact{c_2}$ immediately, else we
would be stuck in $\styleloc{L_1}$ forever.
However, we might not be able to take $\styleact{c_3}$ from $\styleloc{L_2}$ before
the environment sends us back to $\styleloc{L_1}$ with $\styleclock{x} = 0$, where the previous case applies.
Since we have a winning strategy from $\styleloc{L_0}$ when the environment moves with
$\styleclock{x} \leq 2$, we can just wait if it happens, or else win with action $\styleact{c_1}$ when
$\styleclock{x}$ reaches value 2. 
The strategy synthesized by \Uppaal{} is illustrated in \cref{fig:strat_demo:2}.
Notice that $\styleact{c_2}$ winning in $\styleloc{L_1}$ when $\styleclock{x}
= 2$ is merged with applying $\styleact{c_2}$ in $\styleloc{L_1}$ when $1 < \styleclock{x}
< 2$.

Let us focus on how to win from state $(\styleloc{L_1}, x=0)$.
From (3), we wait until $\styleclock{x} > 1$ where (2) can be applied, taking action
$\styleact{c_2}$ as soon as possible. But such a moment does not exist: when
$\styleclock{x} = 1$, it is too early to take $\styleact{c_2}$ ; when $\styleclock{x}
= t > 1$, it is too late as we could have taken $\styleact{c_2}$ earlier at a time
$t'$ such that $1 < t' < t$.
Therefore, no run going through $(\styleloc{L_1}, \styleclock{x} = 0)$ complies
with the strategy provided by \Uppaal{}. Yet, we need to be able to react if the environment
moves from $\styleclock{L_0}$ to $\styleclock{L_1}$ when $\styleclock{x}=0$.

The way we handle this in our strategy is by adding additional information on how long to wait before acting. In~\cref{fig:strat_demo:2} (bottom), we show how we modify (2). This new specification requires the player to commit to an action in advance and wait until the specified time interval before executing it. In $\styleloc{L_1}$ when applying (2a), the action $\styleact{c_2}$ can be taken at any point in time within $1 < \styleclock{x} < 2$. This avoids the issue in \Uppaal{}'s strategy, where an action is always fired at the first available moment (even if it might not exist).

A key subtlety arises: unlike \Uppaal{}'s synthesis, we cannot simply merge (2a) and (2b).  If we did, it would be possible to follow the strategy but end up looping infinitely between $\styleloc{L_1}$ and $\styleloc{L_2}$. Instead, we keep them separate and ensure they remain disjoint. With this change, at most one iteration of the loop can happen, no matter how well the environment plays the uncontrollable actions. We will discuss the exact procedure for synthesizing such a strategy later.
\end{example}

\subsection{Outline}
In \cref{sec:model}, we recall the basic definitions of the model for Parametric
Timed Games.
\Cref{sec:InstrList} introduces our new representation of a controller strategy,
the strategy specification. Modifications to the strategy synthesis algorithm
provide a strategy specification representing a winning strategy for all parameter
values found winning.
In \cref{sec:FTS}, we adapt the model to account for forced transitions. We also
introduce the notion of temporal bounds used in an updated version of our algorithm
to solve Parametric Timed Games with forced transitions.
\Cref{sec:ContSynth} presents an algorithm that synthesizes a controller from a strategy specification
such that the controller, when synchronized with the parametric timed game, limits
its possible runs to a subset of the runs that are coherent with the strategy specification, without restricting the environment actions.
Then, \cref{sec:expe} presents our implementation and the experiments conducted.
Finally, conclusion and perspectives are drawn in \cref{sec:concl}.%

\section{Model of Parametric Timed Games}
\label{sec:model}

A Parametric Timed Game (PTG) is a structure based on timed automata (TA). Like classical automata, it has locations connected by discrete transitions. It also has clocks.
Locations are associated a condition on clock valuations (invariant) that must
be satisfied while staying in the location.
An action in a timed automaton is either to take a discrete transition or to let some
time pass.
Discrete transitions have a guard that must be
satisfied in order to take the transition. In a parametric setting, these conditions
use linear terms over clocks and parameters. Parameters are unspecified but constant during a run. A discrete transition has a subset of clocks which are reset when the transition is taken.

    A \emph{clock valuation} is a function $\val_\ClockSet \in \ClockValSet
    {\ClockSet}$ assigning a non-negative real value to each clock.
    A \emph{parameter valuation} $\val_\ParamSet \in \ParamValSet
    {\ParamSet}$ assigns a non-negative rational value to each parameter.
    A \emph{valuation of a game} is a pair $\val = (\val_\ClockSet,\val_\ParamSet)$. 
    The set of all valuations of the game is denoted
    $\ValSet = \ClockValSet{\ClockSet} \times \ParamValSet{\ParamSet}$.
A \emph{linear term} over $\ParamSet$ is a term defined by the following grammar:
$\ParamLinearTerm \; := \; k \; | \; k\styleparam{p} \; | \; \ParamLinearTerm + \ParamLinearTerm$
where $k \in \mathbb{Q}$ and $\styleparam{p} \in \ParamSet$.
\emph{Zones} allow for capturing a set of valuations in a game.
The set of \emph{parametric zones} $\ZoneSet(\ClockSet,\ParamSet)$ is the set of formulas
defined inductively by the following grammar:
\(\ZoneFormula \; := 
\; \top
\; | \; \ZoneFormula \land \ZoneFormula \; | \; \styleclock{x} \sim \ParamLinearTerm 
\; | \; \styleclock{x}-\styleclock{y} \sim \ParamLinearTerm
\; | \; \ParamLinearTerm' \sim \ParamLinearTerm\)
where $\styleclock{x}, \styleclock{y}\in\ClockSet$, ${\sim} \in \{ <; \leq; =; \geq;
> \}$ and $\ParamLinearTerm$ and $\ParamLinearTerm'$ are linear terms over $\ParamSet$.

In a two-player timed game, discrete transitions are partitioned
between controllable transitions and uncontrollable (environment) transitions.

\begin{definition}[PTG]
A \emph{Parametric Timed Game} is a tuple of the form
$\game = (\LocSet , \ClockSet, \ParamSet, \LabelSet, \TransSet_c, \TransSet_u, \loc_0, \Inv)$ such that
\begin{itemize}[noitemsep, topsep=0pt]
    \item $\LocSet$, $\ClockSet$, $\ParamSet$, $\LabelSet$ are sets of \emph{locations}, \emph{clocks},
        \emph{parameters}, \emph{transition labels}.
    \item
        $\TransSet = \TransSet_c \sqcup \TransSet_u$ is partitioned into sets of \emph{controllable} and \emph{uncontrollable} transitions.
        $\TransSet \subseteq \LocSet \times \ZoneSet(\ClockSet,\ParamSet) \times
            \LabelSet \times \SubsetSet{\ClockSet} \times \LocSet$
        is the set of \emph{transitions} of the form
        $(\styleloc{\loc},\guard,\styleact{a},Y,\styleloc{\loc'})$ where:
        \styleloc{\loc}, \styleloc{\loc'} are source
        and target locations, $\guard$ is the guard,%
\footnote{Note that we extend the definition in~\cite{dahlsenjensen2024ontheflyalgorithmreachabilityparametric}
by allowing diagonal constraints in guards. The main reason is that the synthesized controller will need diagonal constraints.}        
        
        \styleact{a} the label, $Y$ the set of clocks to reset.
    
    \item \styleloc{\loc_0} is the initial location.
    \item $\Inv \; : \; \LocSet \to \ZoneSet(\ClockSet,\ParamSet)$ associates an
        \emph{invariant} with each location.
\end{itemize}
\end{definition}

Function $\val_\ParamSet$ is naturally extended to linear terms on parameters, by
replacing each parameter in the term with its valuation.
With $\val \models \ZoneFormula$, we denote that
valuation $\val = (\val_\ClockSet, \val_\ParamSet)$ \emph{satisfies}
a zone $\ZoneFormula$.
Zones
can also be seen as a convex set in the space of valuations
by considering those satisfying the condition. 

\subsection{Semantics of Parametric Timed Games}

A \emph{state} of a PTG consists of a location and a valuation of clocks and parameters.
Transitions modify clock valuations by letting time pass or resetting clocks.

    Let $\val = (\val_\ClockSet,\val_\ParamSet)$ be a valuation of the game
    and $\delay \geq 0$ a delay.
    $\forall \styleclock{x}\in\ClockSet: (\val_\ClockSet+\delay)(\styleclock{x}) =
    \val_\ClockSet(\styleclock{x}) + \delay$
and
    $\val+\delay = (\val_\ClockSet +\delay , \val_\ParamSet)$.
    Let
    $Y \subseteq \ClockSet$.
    $\val_\ClockSet[Y:=0]$ is the valuation obtained by \emph{resetting the clocks} in $Y$,
    \ie{}:
    $\forall \styleclock{x}\in Y: \val_\ClockSet[Y:=0](\styleclock{x})=0$ and
    $\forall \styleclock{x}\in\ClockSet\setminus Y:
        \val_\ClockSet[Y:=0](\styleclock{x})=\val_\ClockSet(\styleclock{x})$
, and  
    $\val[Y:=0]= (\val_\ClockSet[Y:=0], \val_\ParamSet)$.

The semantics of a Parametric Timed Game is defined
as a timed transition system on states, with timed and discrete transitions, and a winning condition.
A \emph{state} of a PTG is a pair $(\styleloc{\loc},\val)$ where $\styleloc{\loc}$
is a location and $\val$ a valuation of the game satisfying its invariant:
$\val \models\Inv(\styleloc{\loc})$. The \emph{state space} is then
$\StateSpace = \{ (\styleloc{\loc},\val)\in \LocSet \times \ValSet \; | \;
\val \models \Inv(\styleloc{\loc}) \}  = {\bigcup}_{\styleloc{\loc} \in \LocSet}
\; \{\styleloc{\loc}\} \times \Inv(\styleloc{\loc})$.
Let $\delay \in \mathbb{R}_{\geq 0}$ be a time delay. A \emph{timed transition} is
a relation ${\temptrans} \in \StateSpace \times \StateSpace$ s.t.
$\forall (\styleloc{\loc},\val),(\styleloc{\loc'},\val')\in \StateSpace:
(\styleloc{\loc},\val)\temptrans (\styleloc{\loc'},\val')$ iff
$\styleloc{\loc} = \styleloc{\loc'}$ and $\val' = \val + \delay$.
Let $\trans = (\styleloc{\loc},\guard,\styleact{a},Y,\styleloc{\loc'}) \in \TransSet$
be a transition. A \emph{discrete transition} is a relation
${\disctrans} \in \StateSpace \times \StateSpace$ s.t.
$\forall (\styleloc{\loc},\val), (\styleloc{\loc'},\val')\in \StateSpace:
(\styleloc{\loc},\val) \disctrans (\styleloc{\loc'},\val')$ iff
$\val \models \guard$ and $ \val' =  \val[Y:=0]$.

States are grouped in symbolic states, similar to valuations grouped in zones.

\begin{definition}[symbolic state]
    A \emph{ symbolic state} of a PTG is a pair $\SymbState=(\styleloc{\loc},\zone)$ where $\styleloc{\loc}$
    is a location and $\zone$ a zone of the game satisfying its invariant:
    $\zone \models \Inv(\styleloc{\loc})$.
    The symbolic state $(\styleloc{\loc},\zone)$ represents the subset of states $
    \{\styleloc{\loc}\} \times \zone \subseteq \StateSpace$.

    \noindent
    The set of \emph{temporal successors} of $\SymbState$ is defined
    as $\TempSucc{\SymbState}=\{\state' \mid \exists\,\delay,
    \exists \state\in\SymbState: \state \temptrans{} \state'\}$. Similarly, its \emph{temporal predecessors} are defined
    $\TempPred{\SymbState}=\{\state \mid \exists\,\delay,
    \exists \state'\in\SymbState: \state \temptrans{} \state'\}$.
\end{definition}

Let $\vec{0}$ be the clock valuation where all clocks have value $0$.
The set of possible initial states of the PTG is
$\SymbState_0 = \{ (\styleloc{\loc_0}, (\vec{0}, \val_\ParamSet)) \; | \;
\val_\ParamSet \in \ParamValSet{\ParamSet}:
(\vec{0}, \val_\ParamSet) \models \Inv(\styleloc{\loc_0}) \}$.

A \emph{run} of the PTG $G$ is a finite or infinite
sequence of states
$\state_0 \state_1 \state_2 \ldots$ s.t. $\state_0 \in \SymbState_0$ and
$\forall i \in \mathbb{N}, \exists\delay \in \mathbb{R}_{\geq 0}, \exists s\in\StateSpace,\;
\state_{i} \temptrans s \disctrans \state_{i+1}$.
A \emph{history} is a finite run. Given a history $h$, we define $ls(h)$
to be the last state of $h$.
Given a game $\game$, we denote its set of runs by $\RunSet$ and
the histories by $\HistSet$.

Since we deal with reachability games, the objective is specified by a set of locations in the PTG.
A run is winning 
if it reaches one of these locations.
Let $\TargetSet\subseteq\LocSet$ be a \emph{reachability objective}.
The set of \emph{winning runs} $\Omega_{Reach}(\TargetSet)$ is the subset of runs that visit $\TargetSet$:
$\Omega_{Reach}(\TargetSet) = \{ \run \in \RunSet \mid \exists\,\styleloc{\loc}\in \TargetSet, \exists
\val\in\ValSet: (\styleloc{\loc},\val)\in \run\}$.

\subsection{Strategies in Parametric Timed Games}

 A controller strategy $\strat_c$
        models decision-making. Building upon the strategy definition of~\cite{DBLP:conf/adhs/ChatainDL09}, it is a function taking a history and deciding to either wait indefinitely, or to wait for a finite delay and apply a discrete transition.
In the latter case, we will require the transition to be available after the delay
(\ie{} the invariant holds and the guard of the transition is enabled).

\Cref{motiv} exhibited a critical delay interval to select an action.
We now introduce controller strategies that, given a history, return a non-empty set
of valid decisions, one of which will be chosen non-deterministically.

    \begin{definition}[controller strategy]
        A \emph{controller strategy} $\strat$ is a function
        $\strat: \mathcal{H} \to \SubsetSet{(\mathbb{R}_{\geq 0} \times \TransSet_c) \cup \{\infty\}} \setminus \{\emptyset\}$, 
        s.t.\ for all $h \in \mathcal{H}$, if $(\delay, \trans_c) \in \strat(h)$,
         then for some $\state_{\delay}, \state_{t_c} \in\StateSpace$, 
         $ls(h) \to^{\delay} \state_{\delay} \to^{t_c} \state_{t_c}$.
         
         The \emph{delay of a decision} is: 
         $delay(\delta,t)=\delta$ and $delay(\infty)=\infty$.
    \end{definition}

    Note $\infty$ denotes that we ``give up our turn", not necessarily that we stay in this state
    forever. So we do not require that the invariant stays true in this case.

    A controller strategy $\strat$ induces a set of runs coherent with the strategy.
    Either a decision of the controller strategy is applied, or a decision leading to an uncontrollable action occurs with a delay lower or equal to the delay of a controller decision of the strategy. This corresponds to the case where an uncontrollable action intercepts a controller decision of the strategy.
    If both the controller and the environment select the same delay, we consider that the controller cannot guarantee that
    its transition will be taken, thus the environment can intercept.

\begin{definition}[run coherent with controller strategy]
    Run $\run = s_0 s_1 s_2 \ldots$ is coherent with the controller strategy $\strat$ if and only if $\forall i \in \mathbb{N}$:
    \begin{itemize}[noitemsep, topsep=0pt]
        \item Either the run ends in $s_i$, and $\infty \in \strat(s_0 s_1 \ldots s_i)$
        \item Or there exists $ (\delay, \trans_c) \in \strat(s_0 s_1 \ldots s_i)$ and $\state' \in \StateSpace$ s.t.\ 
        $s_i\to^{\delta} s' \to^{\trans_c} s_{i+1}$.
        \item Or there exists $\decision \in \strat(s_0 s_1 \ldots s_i)$, $\delay' \leq delay(\decision) \in \DelaySet$, $\trans_u \in \TransSet_u$ and $s' \in \StateSpace$ such that 
        $s_i\to^{\delay'} s' \to^{\trans_u} s_{i+1}$.
    \end{itemize}
\end{definition}

\begin{definition}[winning strategy]
    A controller strategy $\sigma$ is \emph{winning} w.r.t.\ a reachability objective 
    $\TargetSet$ iff all runs coherent with $\sigma$ are winning w.r.t.\ $\TargetSet$.
\end{definition}

Previous work (\cite{dahlsenjensen2024ontheflyalgorithmreachabilityparametric}) 
introduced an algorithm aimed at solving the following question: 
Given a Parametric Timed Game $\game$ and a
Reachability Objective $\TargetSet$, for which parameter valuations is there a winning controller strategy from the initial state?
We now aim to solve the following questions: \emph{How can we produce a controller strategy that guarantees to win from those parameters valuations? Can we express the controller as a PTA, to be synchronized with the system?}

\section{Strategy Specification and Synthesis}
\label{sec:InstrList}

We now introduce strategy specifications, as a means to specify infinite strategies by finite objects. 
A strategy specification is a list of instructions. Instructions can be of two types: 
First, $( \SymbState, \bot, \Wait)$ where $\SymbState$ is a symbolic state and $\Wait$
is
a signal to wait. This instruction represents the instruction ``When in $\SymbState$,
wait indefinitely''. Second, $(\SymbState, \SymbState', \trans_c)$ where $\SymbState$
and $\SymbState'$ are symbolic states and $\trans_c$ is a controllable transition.
This represents the instruction ``When in $\SymbState$, wait until $\SymbState'$ is
reached, and apply the discrete transition $\trans_c$".

\begin{definition}\label{def:strat:instr}
        A strategy specification $\PartList$ is a list of instructions $i$
        of the form: $i=( \SymbState, \SymbState', \trans_c)$ or $i=( \SymbState,
        \bot, \Wait)$, where $\SymbState$ and
        $\SymbState'$ are symbolic states, $\trans_c$ is a controllable transition
        such that $\trans_c$ is applicable from every state $\state'$ of $\SymbState'$
        and $\SymbState \subseteq \TempPred{\SymbState'}$, $\bot$ is undefined
        and $\Wait$ is a signal to wait.

      The source of an instruction $i$ is denoted by $\src{i}=\SymbState$.
\end{definition}

A state $\state\in\StateSpace$ matches an instruction $i$ of some strategy specification
$\PartList{}$ when $\state \in \src{i}$.
If the current state does not match any instruction of $\PartList{}$,
the default action is to wait indefinitely, as with an instruction of the form $(\SymbState,\bot,\Wait)$.
If the state matches an instruction of the form $(\SymbState, \SymbState', \trans_c)$,
a delay leading to $\SymbState'$ is selected non-deterministically.
From there, transition $\trans_c$ is taken.

\begin{definition}[decisions]
    The set of possible \emph{decisions} specified 
    by instructions from a strategy specification $\PartList$, given current state $\state\in\StateSpace$,
    is defined as follows:
      \begin{align*}
      \dec{(\SymbState, \bot, \Wait),\state} & := \{\infty\} \\
      \dec{(\SymbState, \SymbState', \trans),\state} & := 
	\{ \delta \in \mathbb{R}_{\geq 0} \; | \; \exists \state_\delta \in \SymbState' \text{s.t.\ } 	\state \to^{\delta} \state_\delta \} \times \{\trans\}
	\end{align*}

\noindent The controller strategy $\strat_\PartList$ specified by $\PartList$ on history $\hist\in\HistSet{}$
is defined
as follows:
    \begin{align*}
      \strat_\PartList(\hist) & = \bigcup \{ \dec{i,\LastState(\hist)}
      \mid {i\in \PartList} \textrm{ and }\LastState(\hist) \textrm{ matches } i\}   && \text{if not empty, } \\ 
      \strat_\PartList(\hist) & = \{\infty\} &&\text{otherwise}.
    \end{align*}

\end{definition}

\subsection{Synthesis of Strategy Specification}

An algorithm to compute valuations of the parameters for which a winning strategy
exists, is presented in~\cite{dahlsenjensen2024ontheflyalgorithmreachabilityparametric}.
We quickly recall how this algorithm works, before introducing some modifications
to synthesize a winning strategy specification.
\Cref{alg:main} presents the algorithm, where our modifications are highlighted in
green boxes. The blue boxes will be addressed in \cref{sec:FTS}. It features two main
parts: the exploration ({\sc explore}) and the back-propagation of winning states
({\sc update}).


\begin{algorithm}

  \caption{For PTG $\game = (\LocSet, \ClockSet, \ParamSet, \LabelSet, \TransSet_c,
  \TransSet_u, \styleloc{\loc_0}, \Inv)$ and reachability objective $\TargetSet$,
  returns the set of all parameter valuations that win the game.
  \label{alg:main}}
  \begin{algorithmic}[1]
    \State $\Explored, \WaitingUpdate, \WaitingExplore \gets \emptyset, \emptyset, \{\TempSucc{\SymbState_0}\}$
    \Comment{Symbolic state sets}
    \State $\mathit{Win}[\,], \Depends[\,]$
      \ColorString{ColorCovPrune}{$, \ForcedMoves[\,]$} $\gets \emptyset, \emptyset,
      \emptyset$
    \Comment {Maps from symbolic states}
    \ColorBox{ColorLosingProp}{
    	\State $\PartListAlg := [\,]$
    	\Comment {List of triplets (Symbolic state, Symbolic state, action)}
    }
    \State $\WinningParam := \mathtt{False}$

    \medskip
    \Function{solvePTG}{}
    \While{$\WaitingExplore \neq \emptyset \lor \WaitingUpdate \neq \emptyset$
    } \label{alg:line:mainloop}
    \State Choose either \Call{explore}{ \hspace*{-1mm}} or  \Call{update}{ \hspace*{-1mm}}
    \label{alg:line:functionCall}
    \EndWhile
    \State \Return $(\WinningParam$\ColorString{ColorLosingProp}{$, \PartListAlg$})
    \EndFunction

    \medskip
    \Procedure{explore}{}
    \State Pick $\SymbState$ from $\WaitingExplore$ \label{alg:line:explore:pick}
    \For{\trans transition from \SymbState} \label{alg:line:explore:forbegin}
    \State $\SymbState' := \TempSucc{Succ(\trans, \SymbState)}$
    \State $\Depends[\SymbState'] \gets \Depends[\SymbState'] \cup \{\SymbState\}$
    \If{$\SymbState' \not\in \Explored$}
     $\WaitingExplore \gets \WaitingExplore \cup \{\SymbState'\}$
    \EndIf
    \EndFor \label{alg:line:explore:forend}
    \If{$\SymbState.\loc \in \TargetSet$} \label{alg:line:explore:beginif}
    \State $\mathit{Win}[\SymbState] \gets \SymbState$, $\WaitingUpdate \gets \WaitingUpdate \cup \Depends[\SymbState]$\label{alg:line:explore:winupdate}
    
    \ColorBox{ColorLosingProp}{
    	\State $\PartListAlg \gets \PartListAlg \cup \{(\SymbState, \bot, \Wait)\}$
        \label{alg:line:explore:instrlist}
    }
    \EndIf\label{alg:line:explore:endif}
    \ColorBox{ColorCovPrune}{
    	\State $\Uncontrollable := 
    		\bigcup\{\guardFunction{\trans_u} \mid \trans_u \in \TransSet_u \,\land\,
        \trans_u \text{ transition from } \SymbState\}$\label{alg:line:explore:ft:Uguard}
    	\State $\Controllable := 
        \bigcup\{\guardFunction{\trans_c} \mid \trans_c \in \TransSet_c \,\land\,
        \trans_c \text{ transition from } \SymbState\}$\label{alg:line:explore:ft:Cguard}
    	\State $(\InvBound_{\mathit{In}}, \InvBound_{\mathit{Out}}) := \UTempSplit(\Inv
      (\SymbState.\loc))$
        \label{alg:line:explore:ft:split}
    	\State $\ForcedMoves[\SymbState] \gets (\InvBound_{\mathit{In}} \cap \Uncontrollable)
    		\setminus \Controllable$\label{alg:line:explore:ft:fmin}
    	\State $\ForcedMoves[\SymbState] \gets \ForcedMoves[\SymbState] \cup (\InvBound_{\mathit{Out}} \cap \overline{\Uncontrollable}) \setminus \overline{\Controllable}$
        \label{alg:line:explore:ft:fmout}
	}
    \State $\WaitingUpdate \gets \WaitingUpdate \cup \{\SymbState\}$ \label{alg:line:explore:toupdate}, $\Explored \gets \Explored \cup \{\SymbState\}$ \label{alg:line:explore:explored}
    \EndProcedure

    \medskip
    \Procedure{update}{}
    \State Pick $\SymbState$ from $\WaitingUpdate$
    \State $\Uncontrollable := 
      \bigcup\{
    		\Pred(t_u,\SymbState' \setminus \mathit{Win}[\SymbState'])\mid 
        \trans_u \in \TransSet_u \,\land\,\SymbState' =\TempSucc{Succ(\trans_u, \SymbState)}\}$\label{alg:line:update:uncontrol}
    \For{\trans controllable transition from \SymbState}
      \State $\WinningMove := 
      \Pred(\trans, \mathit{Win}[\TempSucc{Succ(\trans, \SymbState)}])$
        \ColorString{ColorLosingProp}{$\setminus\, \Uncontrollable$}
			\Comment{Union of zones} \label{alg:line:update:winmove}
    	\For{$\SymbState_i \in \WinningMove$}
		    \State $\NewWin_i := \SafePred( \SymbState_i, \Uncontrollable) \:
					     \cap \: \SymbState$ \label{alg:line:update:safepred}
        \ColorBox{ColorLosingProp}{
          \For{$(\SymbState_{\mathit{solved}},\_,\_) \in \PartListAlg$}
            \label{alg:line:update:begininstr}
            \State $\NewWin_i \gets \NewWin_i \setminus \SymbState_{\mathit{solved}}$
          \EndFor\label{alg:line:update:endinstr}
          \State $\PartListAlg \gets \PartListAlg \cup
              (\NewWin_i, \SymbState_i, \trans)$\label{alg:line:update:instrcontrol}\label{alg:line:update:pair1a}
        }
        \State $\NewWin \gets \NewWin \cup \NewWin_i$\label{alg:line:update:pair1b}
  		\EndFor
  	\EndFor
    \ColorBox{ColorCovPrune}{ 
		\For{$\SymbState_i$ zone in the union of zones $\ForcedMoves[\SymbState]$}
      \label{alg:line:update:ft:begin}
			\State $\NewWin_i := \SafePred( \SymbState_i, \Uncontrollable) \:  \cap
				\: \SymbState$
			\For{$(\SymbState_{\mathit{solved}},\_,\_) \in \PartListAlg$}\label{alg:line:update:remove_overlapsA}
				\State $\NewWin_i \gets \NewWin_i \setminus \SymbState_{\mathit{solved}}$
			\EndFor\label{alg:line:update:remove_overlapsB}
      \State $\NewWin \gets \NewWin \cup \NewWin_i$\label{alg:line:update:pair2a}
			\State $\PartListAlg \gets \PartListAlg \cup (\NewWin_i, \bot, \Wait)$\label{alg:line:update:pair2b}
      \label{alg:line:update:ft:end}
		\EndFor
	}
    \If{$\NewWin \not\subseteq \mathit{Win}[\SymbState]$}\label{alg:line:update:startend}
    	\State $\WaitingUpdate \gets \WaitingUpdate \cup \Depends[\SymbState]$, $\mathit{Win}[\SymbState] \gets \mathit{Win}[\SymbState] \cup \NewWin$\label{alg:line:update:addnewwin}
    \EndIf\label{alg:line:update:finishend}
    \State $\WinningParam \gets \ProjectParam{(\mathit{Win}[\SymbState] \cap
    	\SymbState_0)}$ \label{line:report}
    \EndProcedure
  \end{algorithmic}
\smallskip
\hrule
\smallskip
  {\scriptsize Green boxes are additions to the original algorithm \cite{dahlsenjensen2024ontheflyalgorithmreachabilityparametric}
  for generating a strategy specification. Blue boxes denote additions for forced transition semantics, explained in \cref{sec:FTS}.}
\end{algorithm}

\medskip
\noindent{\bf Original algorithm from~\cite{dahlsenjensen2024ontheflyalgorithmreachabilityparametric}.}
The overall idea is to compute reachable states while propagating winning conditions backward based on controllable and uncontrollable transitions. Starting from the initial symbolic state $\SymbState{}_0$, as long as some states
remain to be explored or updated, they are handled by the appropriate function (\cref{alg:line:mainloop,alg:line:functionCall}).

Procedure \textsc{explore} chooses a symbolic state $\SymbState{}$ in the set of states to explore $\WaitingExplore{}$ and removes it from this set (\cref{alg:line:explore:pick}).
It then considers one by one all discrete transitions that can be taken from $\SymbState{}$
(\crefrange{alg:line:explore:forbegin}{alg:line:explore:forend}). The successor symbolic
states $\SymbState{}'$ are generated by computing the states reachable from $\SymbState{}$ via the transition ($Succ$) and letting time pass.
$\Depends$ records $\SymbState{}$ as a predecessor of $\SymbState{}'$. $\SymbState{}'$ is added to the set of states to explore, if not already explored yet.
If the location of $\SymbState{}$ is a target, it is considered winning, and its
predecessors added to the set waiting for a back-propagation update (\crefrange{alg:line:explore:beginif}
{alg:line:explore:endif}). In any case, $\SymbState{}$ is also added to that set
(\cref{alg:line:explore:toupdate}),
and declared explored (\cref{alg:line:explore:explored}).

Procedure \textsc{update} performs the back-propagation of winning zones. It operates
on a symbolic state $\SymbState{}$ picked from $\WaitingUpdate{}$.
First, the predecessors
(\Pred) of successors of $\SymbState{}$ by an uncontrollable transition that are
 not already
known as winning are stored in a set $\Uncontrollable{}$
(\cref{alg:line:update:uncontrol}). The predecessors
of the winning part of the successors of $\SymbState{}$ via controllable actions outside $\Uncontrollable{}$
are considered as winning
(\cref{alg:line:update:winmove}).
For such states, 
the safe temporal predecessors (\SafePred) are computed, \ie{} the temporal predecessors
that avoid reaching
a state in $\Uncontrollable$ (\cref{alg:line:update:safepred}). These are the states
from which
the controller can win in $\SymbState$. Finally, if the found winning part ($\NewWin$) is different from the one stored in $\Win[\SymbState]$, it is updated with the new value and a back propagation is initiated by adding the predecessors of $\SymbState$ to $\WaitingUpdate$
(\crefrange{alg:line:update:startend}{alg:line:update:finishend}). The projection
of the zone on its parameters valid in the initial state is reported (\cref{line:report}).

\medskip
\noindent{\bf Modifications for generating the strategy specification.}
In this paper, we not only report the parameters but also construct the winning strategy. \cite{DBLP:conf/wodes/JovanovicFLR12} and \cite{DBLP:journals/ijcon/JovanovicLR19} also discuss how to do this, however their constructed strategies are equivalent to those of \Uppaal (with the addition of parameters) that we compare with in \cref{sec:intro}.
Procedure \textsc{explore} constructs at \cref{alg:line:explore:instrlist} the list
of instructions in the strategy, as defined in \cref{def:strat:instr}, in the case
of a target state reached, and thus time can elapse. The other case, where a controllable
transition must be taken from a state, is handled by the \textsc{update} procedure
at \crefrange{alg:line:update:begininstr}{alg:line:update:endinstr}: states that are already present in the instructions list cannot be considered as
newly winning. The remaining lead to an instruction in the strategy specification
(\cref{alg:line:update:instrcontrol}).

We now state the correctness of the algorithm, as proved in \cref{app:proofs3}.

\begin{restatable}[correct strategy]{theorem}{strategyWinning}
  For all winning initial states, the strategy associated with the specification
  synthesized by \cref{alg:main}
  is winning.
\end{restatable}

\begin{example}
\begin{figure}[b]
  \begin{subfigure}{.4\textwidth}
  \centering
    \begin{tikzpicture}[>=stealth', node distance=2.5cm,initial text=]
        \node[location,initial above,label={left:\styleloc{L_0}}] (l0) {};
        \node[location,label={left:\styleloc{L_1}},below of=l0,
            node distance=1cm] (l1) {};
        \node[location,label={right:\styleloc{Lose}},right of=l1] (lose) {};
        \node[location,accepting,label={right:\styleloc{Win}},above of=lose, node
        distance=1cm](win) {};

        \draw[->] (l0) -- node [left] 
        {$\substack{\styleclock{x} > 1\\
            \styleact{c_1}
        }$} (l1);
        \draw[->,dashed] (l1) -- node 
        {$\substack{\styleclock{x} < \styleparam{p} \\
            \styleact{u_1}
        }$} (lose);

        \draw[->] (l1) -- node [above] 
        {$\substack{\styleact{c_2}
        }$} (win);
    \end{tikzpicture}
    \caption{Example PTG
    \label{fig:generation_ex_1}}
\end{subfigure}
\hfill
\begin{subfigure}{.55\textwidth}
      \centering
      \begin{adjustbox}{width=1\textwidth}
      \begin{tabular}{l  l@{\hspace{2em}} l@{\hspace{2em}} l}
          \toprule
          \textbf{Location} & \textbf{Condition} & \textbf{Wait Until} & \textbf{Action} \\
          \midrule
          \styleloc{L_0} & $0\leq \styleparam{p}$ & $0\leq \styleparam{p} \leq \styleclock{x} \land 1 < \styleclock{x}$ & $\styleact{c_1}$ to reach \styleloc{L_1}. \\
          \midrule
          \styleloc{L_1} & 
          $0\leq \styleparam{p} \leq \styleclock{x} \land 1 < \styleclock{x}$
          & $0\leq \styleparam{p} \leq \styleclock{x} \land 1 < \styleclock{x}$ & $\styleact{c_2}$ to reach \styleloc{Win} \\
          \bottomrule
      \end{tabular}
    \end{adjustbox}
  
  \caption{Strategy to Ensure Controller Win}
  \label{tab:strategy}
\end{subfigure}
\caption{Example PTG and a corresponding strategy generated by the algorithm}
\end{figure}

Applying the algorithm on the game of \cref{fig:generation_ex_1}
returns $\styleparam{p} \geq 0$ as well as the strategy depicted in \cref{tab:strategy}.
The strategy is exactly as expected: the player should take $\styleact{c_1}$ to $
\styleloc{L_1}$ when $\styleclock{x} \geq \styleparam{p}$ (as well as when $\styleclock{x} > 1$, imposed by the guard) and then $\styleact{c_2}$ to $\styleloc{Win}$ immediately. 
\end{example}

\section{Forced Transition Semantic}
\label{sec:FTS}
In~\cite{OTF-TG}, the environment may remain in a location with an invariant when an action is available the moment the invariant breaks --- effectively allowing the environment to violate the invariant by idling. While this is a valid choice, it is not very natural and disallowed in practice; e.g., \Uppaal Tiga \cite{DBLP:conf/cav/BehrmannCDFLL07} enforces that the environment takes an available action (if no controllable action exists) when an invariant is about to break. 
We introduce a forced transition semantics to formalize this behavior, thus justifying the practice in~\cite{DBLP:conf/cav/BehrmannCDFLL07}.

\begin{definition}[forced transition constraint]
    A run $\run$ satisfies the \emph{forced transition constraint} if and only if
    one of the following conditions holds:
    \begin{itemize}[noitemsep, topsep=0pt]
        \item $\run$ is an infinite sequence.
        \item $\run$ is a finite sequence and for all $\delay \in \DelaySet$, there exists $\state_\delay\in \StateSpace$ such that $\LastState(\run) \temptrans \state_\delay$ (No further discrete transition and no invariant break).
        \item $\run$ is a finite sequence and there exists a delay $\delay \in \DelaySet$ such that for all $\delay' \geq \delay$ and $\state_{\delay'} \in \StateSpace$ such that $\LastState(\run) \to^{\delta'} \state_{\delay'}$,  a \emph{controllable} action is available in $\state_{\delay'}$ or no \emph{uncontrollable} action is available in $\state_{\delay'}$. 
    \end{itemize}
\end{definition}

A 
strategy is winning under the forced transition constraint if all its coherent runs satisfying the forced transition constraint are winning.
A state is winning under the forced transition constraint if it has a winning controller strategy.

\subsection{Temporal Bounds}

To solve a PTG with forced transitions, we describe when a transition can
be forced. This happens at the temporal upper bound of the invariant. We also distinguish if this bound intersects the invariant or not.
To compute the temporal upper bound of a zone, we distinguish
how upper and lower temporal constraints restrict time elapsing. Diagonal constraints do not restrict time elapse.

\begin{definition}
    \emph{Upper} (resp. lower) \emph{temporal constraints} $\UpTempCons$ 
    of a zone $\zone = \bigwedge \ZoneFormula_i$ are
    constraints $\ZoneFormula_i$ of type $ \styleclock{x} \sim \ParamLinearTerm $ with ${\sim} \in \{<, \leq, = \}$ (resp. ${\sim} \in \{>, \geq, = \}$). 
    Furthermore, all zones have implicit temporal lower constraints $\styleclock{x}
    \geq 0$ for all $\styleclock{x} \in \ClockSet$.
    The set of upper (resp. lower) temporal constraints of $\zone$ is denoted by
    $\UpTempConsSet(\zone)$  (resp. $\LoTempConsSet(\zone)$).
    All other constraints are called \emph{diagonal constraints}.
\end{definition}

Note that the temporal upper bound of a zone may not be included into the zone itself:
In \cref{fig:FT}, both invariant $\styleclock{x} \leq \styleparam{p}$ and $\styleclock{x} < \styleparam{p}$ have the same temporal upper bound $\styleclock{x} = \styleparam{p}$ while only $\styleclock{x} \leq \styleparam{p}$ contains it.

\begin{definition}
    Let zone $\zone$ consist of the set of constraints $\ZoneFormula$. The
    upper (resp. lower) \emph{temporal closure} of $\zone$, denoted by
    $\UpTempClos{\zone}$ (resp. $\LoTempClos{\zone}$) 
    keeps all diagonal constraints unchanged, and 
    applies the
    following transformations to $\ZoneFormula$:
    \begin{itemize}[noitemsep, topsep=0pt]
        \item Make all upper (resp. lower) temporal constraints of $\ZoneFormula$ non-strict.
        \item Make all lower (resp. upper) temporal constraints of $\ZoneFormula$ strict.
    \end{itemize}
\end{definition}

Let $\zone$ be a zone with an upper temporal constraint $\UpTempCons = \styleclock{x} \sim \ParamLinearTerm$. Then 
$\zone$ gets \emph{invalidated} due to $\UpTempCons$ after time
$\styleclock{x} = \ParamLinearTerm$.

\begin{definition}
    Let $\zone$ be a zone and $\UpTempCons = \styleclock{x} \sim \ParamLinearTerm$
    an upper (resp. lower) temporal constraint of $\zone$.
    The \emph{temporal} upper (resp. lower) \emph{bound} associated with $\UpTempCons$
    and denoted $\UpTempBound{\zone}{ \UpTempCons}$
    (resp. $\LoTempBound{\zone}{\LoTempCons}$) is the zone obtained by replacing $
    \styleclock{x} \leq \ParamLinearTerm$ (resp. $\styleclock{x} \geq \ParamLinearTerm$)
    by $\styleclock{x} = \ParamLinearTerm$ in $\zone$ if $\UpTempCons$ is non-strict, or in $\UpTempClos{\zone}$ otherwise.
\end{definition}

Note that the bound associated with a strict constraint $\styleclock{x}
< \ParamLinearTerm$ is necessarily outside the zone $\zone$, since it satisfies
$\styleclock{x} = \ParamLinearTerm$, whereas for a non-strict constraint $\styleclock{x}
\sim \ParamLinearTerm$ it is necessarily within $\zone$.
The temporal bounds of a zone are the union of the temporal bounds of its constraints.
The proof of \cref{th:2} is in \cref{app:proofs4}.

\begin{definition} 
    Let $\zone$ be a zone. The upper (resp. lower) \emph{temporal bound} of $\zone$
    is the union of the upper (resp. lower) temporal bounds of its temporal upper
    (resp. lower) constraints $\UpTempCons$.

\hfill
    $\UpTempGlobBound{\zone} = \underset{\UpTempCons \in \UpTempConsSet(\zone)}{\bigcup} \UpTempBound{\zone}{\UpTempCons}$.
\hfill
    $\LoTempGlobBound{\zone} = \underset{\LoTempCons \in \LoTempConsSet(\zone)}{\bigcup} \LoTempBound{\zone}{\LoTempCons}$.
\hfill
\mbox{}
\end{definition}
\begin{restatable}{theorem}{temporalBoundZone}
\label{th:2}
    The upper \emph{(lower)} temporal bound of zone $\zone$ corresponds to:
\begin{itemize}[noitemsep, topsep=0pt]
\item    $\UpTempGlobBound{\zone}=\{ \val + \delay_{\mathit{sup}} \mid \val \in \zone \text{
    such that } \delay_{\mathit{sup}} =  \sup \{ \delay \in \mathbb{R}_{\geq 0} \mid \val + \delay \in \zone \} \in \mathbb{R}_{\geq 0} \}$.
\item $\LoTempGlobBound{\zone}=\{ \val - \delay_{\mathit{sup}} \mid \val \in \zone \text{
    such that } \; \delay_{\mathit{sup}} =  \sup \{ \delay \in \mathbb{R}_{\geq 0} \mid \val - \delay \in \zone \} \in \mathbb{R}_{\geq 0} \}$.
\end{itemize}
\end{restatable}

\subsection{Algorithm Handling Forced Transitions}
To solve the PTG, we add so-called
$\ForcedMoves$ from states where the invariant breaks and only uncontrollable transitions are available. The controller
can use $\ForcedMoves$ to claim the victory, since the 
environment didn't respect the forced-move semantics.
We use the temporal upper bounds to make this precise.

\begin{example}
    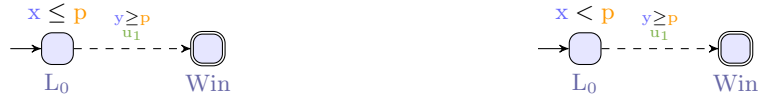
\begin{figure}[b]
    \vspace{-2em}
    \begin{subfigure}{.4\textwidth}\centering
          \begin{tikzpicture}[>=stealth', node distance=3cm,initial text=]
              \node[location,initial left,label={below:\styleloc{L_0}},label={above:$\styleclock{x} \leq \styleparam{p}$}] (l0) {};
              \node[location,accepting,label={below:\styleloc{Win}},right of=l0, node distance=2cm] 
              (win) {};
              \draw[->,dashed] (l0) -- node [above] 
              {$\substack{\styleclock{y} \geq \styleparam{p} \\
                  \styleact{u_1}
              }$} (win);
          \end{tikzpicture}
          \caption{PTG with a forced transition
          \label{fig:FT1}}
      \end{subfigure}
      \hfill
      \begin{subfigure}{.45\textwidth}\centering
        \begin{tikzpicture}[>=stealth', node distance=3cm,initial text=]
            \node[location,initial left,label={below:\styleloc{L_0}},label={above:$\styleclock{x} < \styleparam{p}$}] (l0) {};
            \node[location,accepting,label={below:\styleloc{Win}},right of=l0, node distance=2cm] 
            (win) {};
            \draw[->,dashed] (l0) -- node [above] 
            {$\substack{\styleclock{y} \geq \styleparam{p} \\
                \styleact{u_1}
            }$} (win);
        \end{tikzpicture}
        \caption{PTG without forced transitions
        \label{fig:FT2}}
    \end{subfigure}
      \caption{Examples of PTG with and without forced transitions}
      \label{fig:FT}
      \end{figure}
    
    \Cref{fig:FT} shows two PTGs with initial state $( \styleloc{L_0}, \styleclock{x}=\styleclock{y}=0)$. In both cases,
    the temporal upper bound of the invariant is $\styleclock{x} = \styleparam{p}$.
    In \cref{sub@fig:FT1}, this bound can be reached, so the uncontrollable transition $\styleact{u_1}$ can be forced, and the game is winning. 
    In \cref{sub@fig:FT2}, this bound cannot be reached, so $\styleact{u_1}$ cannot be forced, and the game is losing.
In \cref{sub@fig:FT1}, since the constraint $\styleclock{x} \leq \styleparam{p}$ is non-strict, we intersect the bound directly with the guard of $\styleact{u_1}$. We get $\ForcedMoves$ equals $\styleclock{x}=\styleparam{p} \land \styleclock{y} \geq \styleparam{p}$.
In \cref{sub@fig:FT2}, since the constraint $\styleclock{x} < \styleparam{p}$ is strict, we intersect the bound of the invariant $\styleclock{x} = \styleparam{p}$ with the \emph{temporal upper closure} $\styleclock{y} > \styleparam{p}$ of the guard of $\styleact{u_1}$. We get $\ForcedMoves$ equals $\styleclock{x}=\styleparam{p} \land \styleclock{y} > \styleparam{p}$ (which cannot be reached).
\end{example}

The blue boxes in \cref{alg:main} handle the forced transitions.
During the exploration (\textsc{explore} in \cref{alg:main}), $\ForcedMoves$
are collected for the processed symbolic state $\SymbState{}$. First, at \cref{alg:line:explore:ft:Uguard,alg:line:explore:ft:Cguard},
the guards associated with controllable and uncontrollable transitions available in
$\SymbState{}$ are collected. Then, function $\UTempSplit{}(\Inv(\SymbState.\loc))$
(\cref{alg:line:explore:ft:split}) computes the upper temporal bounds, grouped as
being inside the zone or outside it. For both of these sets, we save the part that has available uncontrollable transitions and no controllable transitions in $\ForcedMoves[\SymbState]$ (\cref{alg:line:explore:ft:fmin},\cref{alg:line:explore:ft:fmout}). The computation with the temporal bounds outside the zone is done with the closed (un)controllable sets (\cref{alg:line:explore:ft:fmout}).
In procedure \textsc{update}, at \crefrange{alg:line:update:ft:begin}{alg:line:update:ft:end},
the forced moves do not correspond to controllable transitions.
Therefore, the strategy of the controller must be to wait until the environment is
forced to move.

\section{Controller Synthesis}
\label{sec:ContSynth}

We now synthesize a controller automaton that
enforces a given winning strategy, without restricting uncontrollable actions. A method for constructing a controller from a strategy specification was previously proposed in~\cite{DBLP:conf/formats/DavidFLZ14}. However, it is based on the \Uppaal strategy semantics, so it exhibits the same issues described in~\cref{sec:intro}. We take inspiration from the proposed structure adapting it to our 
strategy specification, using a new concept of $\epsilon$-bounds.
We assume that for all $\state\in\StateSet{}$ there exists at most one transition labelled
$a\in\LabelSet$.

\subsection{Controller synchronization}

A controller is a PTA with its own locations, clocks and parameters. It may read the values of the 
clocks and parameters with the PTG it controls, but it may not change
them (\ie{} not reset clocks).
When synchronised with the original PTG, it will thus restrict its behaviour.
The controller has its own discrete transitions that may share a label with transitions of the controlled PTA. 

\begin{definition}[controller]
    A \emph{controller} $\cont$ for a Parametric Timed Game
    $\game = (\LocSet^{\game} , \ClockSet^{\game}, \ParamSet^{\game}, \LabelSet^{\game},
    \TransSet_c^{\game}, \TransSet_u^{\game}, \loc_0^{\game}, \Inv^{\game})$, with
    $\TransSet{}^{\game}=\TransSet_c^{\game} \sqcup \TransSet_u^{\game}$,
    is a PTA $(\ClockSet^{\cont}, \ParamSet^{\cont}, \LocSet^{\cont}, \LabelSet^{\cont}, \TransSet^{\cont}, \loc^{\cont}_0)$ such that $\cont$ has no clock reset on $\ClockSet^{\cont} \cap \ClockSet^{\game}$.
\end{definition}

In the parallel composition of the controller PTA and the controlled PTG, a transition
with a shared label $a$ 
occurs simultaneously in both automata.
Transitions with a non-shared label
occur independently.

\begin{definition}
    The \emph{parallel composition}
    $\ParallelComp{\cont}{\game}$ 
    of a controller and its PTG is a labelled transition system where:
    \begin{itemize}[noitemsep, topsep=0pt]
    \item The set of states $\StateSpace^{\ParallelComp{\cont}{\game}}$ is the subset
    of
    $
    \LocSet^{\cont} \times 
    \LocSet^{\game} \times 
    \ValSet(
        \ClockSet^{\cont} \cup \ClockSet^{\game},
        \ParamSet^{\cont} \cup \ParamSet^{\game}
    )$ such that,
    for all $(\loc^{\cont} ,\loc^{\game}, \val) \in \StateSpace^{\ParallelComp{\cont}
    {\game}}$, $\val\models\Inv(\loc^{\game}) \land \Inv(\loc^{\cont})$.
    
    \item It is equipped with 3 types of transitions:
    \begin{itemize}
        \item[-] Timed transitions:
            for $\delay \in \DelaySet, \; (\loc^{\cont} ,\loc^{\game},
            \val ) \temptrans (\loc^{\cont} ,\loc^{\game}, \val' )$
            iff $\val' = \val + \delay$.

        \item[-] Internal discrete transitions:
            For $a \in \LabelSet^{\cont} \setminus \LabelSet^{\game}$,
                $(\loc^{\cont}_0, \loc^{\game}_0, \val) \to^{a}
                    (\loc^{\cont}_1, \loc^{\game}_0, \val')$ iff 
                    $\exists\trans^{\cont} = (\loc^{\cont}_0, \guard, a, Y, \loc^
                        {\cont}_1)\in \TransSet^{\cont}$
                    such that $\val\models\guard$ and
                        $\val' = \val [Y:=0]$.
            For $a \in \LabelSet^{\game} \setminus \LabelSet^{\cont}$,
                $(\loc^{\cont}_0, \loc^{\game}_0, \val) \to^{a}
                    (\loc^{\cont}_0, \loc^{\game}_1, \val')$ iff 
                    $\exists\trans^{\game} = (\loc^{\game}_0, \guard, a, Y, \loc^
                        {\game}_1)\in \TransSet^{\game}$
                    such that $\val\models\guard$ and
                        $\val' = \val [Y:=0]$.

        \item[-] Parallel discrete transitions:
            For $a \in \LabelSet^{\cont} \cap \LabelSet^{\game}$,
            $(\loc^{\cont}_0 ,\loc^{\game}_0, \val ) \to^{a}
                (\loc^{\cont}_1 ,\loc^{\game}_1, \val' )$ iff
                $\exists\trans^{\cont} = (\loc^{\cont}_0, \guard^{\cont}, a,
                    Y^{\cont}, \loc^{\cont}_1)\in \TransSet^{\cont} \land
                    \exists\trans^{\game} = (\loc^{\game}_0, \guard^{\game}, a,
                        Y^{\game}, \loc^{\game}_1) \in \TransSet^{\game}$
                such that $\val\models\guard^{\cont} \land \guard^{\game}$, and 
                    $\val' = \val[Y^{\cont} \cup Y^{\game}:=0]$.
    \end{itemize}
\end{itemize}
\end{definition}

\subsection{$\epsilon$-Lower Temporal Bound}

For an instruction $(\SymbState_1, \SymbState_2, \trans)$ the controller should force
the transition $\trans$ to occur in $\SymbState_2$.
We use an upper bound of $\SymbState_2$ as an invariant to force the transition to
happen before leaving $\SymbState_2$. In case $\SymbState_2$ does not have an upper
bound, we create one, modelled using a new parameter $\styleparam{\epsilon}$: the
transition should then be taken within an $\epsilon$ delay after $\SymbState_2$ is
reached.

\begin{definition}
    Let $\zone$ be a zone in $\mathcal{Z}(\ClockSet,\ParamSet)$ represented by the formula $\ZoneFormula$ and $\LoTempCons = x \sim plt$ a lower temporal constraint of $\zone$.
   The \emph{$\epsilon$-temporal lower bound} $\LoTempBoundEpsilon{\zone}{\LoTempCons}$
   is the zone in $\mathcal{Z}(\ClockSet, \ParamSet \cup \{\epsilon\})$ defined as
   $\phi \land x \leq plt + \epsilon$.
\end{definition}

Function \textsc{AddEpsilonBounds} in \cref{alg:preprocess} describes the modification
of the strategy by
adding instructions
for zones without upper temporal constraints.
For an instruction $(\SymbState_1, \SymbState_2, \trans_c)$ such that $\SymbState_2$
has no upper temporal constraint (condition on \cref{alg:line:pre:noupper}),
the instruction is split between the $\epsilon$-lower temporal bounds of
the zone of $\SymbState_2$.
For every $\epsilon$-lower temporal bound $\LoTempBoundEpsilon{\zone}{\LoTempCons}$
(their union is computed by $\mathit{EpsilonLTempBound}$ at
\cref{alg:line:pre:bounds}),
a new instruction is created from the intersection of its predecessors with $\SymbState_1$
(\cref{alg:line:pre:addepsilon}).
When $\SymbState_2$ can be entered directly in a state located after a delay $\styleparam{\epsilon}$
from its lower bound, an instruction to take the transition immediately (if not interrupted
by the environment) is added (\cref{alg:line:pre:immediate}).

\subsection{Synthesizing the controller}

Function \textsc{ControllerSynthesis} in \cref{alg:ContSynt} considers all strategy
instructions one by one from the augmented list. All actions are initially
part of the controller (if not created by the algorithm, their synchronisation
is not possible and thus they are deactivated in the game).
In \cref{alg:line:synth:urgent1}, if the location of
the instruction does not exist in the controller yet, it is created.
This location is made urgent, \ie{} a location where a discrete action
should be taken without letting time elapse.
Then a location $q_i$ is created that can only be reached from $\SymbState_1$ (\crefrange{alg:line:synth:instrq}{alg:line:synth:instrt}).
When $\decision$ is a controllable action, the invariant of the new state $q_i$  matches
the predecessors of the symbolic state that should be reached according to the strategy,
and the target location is created if necessary, as well as the transition for applying
the strategy (\crefrange{alg:line:synth:instrinv} {alg:line:synth:instract}).
If $\SymbState{}_1$ is included in $\SymbState{}_2$ and $\SymbState{}_2$ has no temporal upper bound, then the instruction was added by the pre-processing (\cref{alg:line:pre:immediate}) where the associated transition should be taken immediately. Thus, location $q_i$ is made urgent
(\crefrange{alg:line:synth:noup}{alg:line:synth:noupurgent}).
Finally, all uncontrollable transitions from $\SymbState{}.\loc$ are inserted in the
controller to allow environment moves (\crefrange{alg:line:synth:uncont}{alg:line:synth:uncontt}).

\begin{algorithm}[h!]
    \caption{Algorithm for controller synthesis from a strategy list. Contains the method to force transition with no upper temporal bound as well} \label{alg:ContSynt}\label{alg:preprocess}
    \begin{algorithmic}[1]
    \Function{AddEpsilonBounds}{\PartListAlg}
    \State $\PartListAlg' := \emptyset$
    \For {$(\SymbState_1, \SymbState_2, \decision) \in \PartListAlg$}
        \If{ ($\decision \neq \Wait$) and $\lnot \mathit{HasUpperTempCons}
        (\SymbState_2.\zone)$}\label{alg:line:pre:noupper}
            \For { $\SymbState_2'$ in the union of zones $\mathit{EpsilonLTempBound}
            (\SymbState_2)$}\label{alg:line:pre:bounds}
                \State $\PartListAlg' \gets \PartListAlg' \cup \{(\SymbState_1 \cap
                (\TempPred{\SymbState_2'}), \SymbState_2', \decision)\}$
                \label{alg:line:pre:addepsilon}
            \EndFor
            \If{$\SymbState_1 \cap \SymbState_2 \neq \emptyset$}
            \State $\PartListAlg' \gets \PartListAlg' \cup \{(\SymbState_1 \cap \SymbState_2,
            \SymbState_2, \decision)\}$
            \label{alg:line:pre:immediate}
            \EndIf 
        \Else \ $\PartListAlg' \gets \PartListAlg' \cup \{(\SymbState_1, \SymbState_2,
        \decision)\}$
            \label{alg:line:pre:copy}
        \EndIf
    \EndFor
    \State \Return $\PartListAlg'$
    \EndFunction

    \medskip
    \Function{ControllerSynthesis}{\PartListAlg} 
    \State $\LabelSet^\cont := \LabelSet^\game$
    \For {$i = (\SymbState_1, \SymbState_2, \decision) \in \Call{AddEpsilonBounds}
    {\PartListAlg}$}
        \State \textbf{if} $\ControllerLoc_{\SymbState_1.\loc}$ does not exist
            \textbf{then} create urgent $\ControllerLoc_{\SymbState_1.\loc}$ 
            \label{alg:line:synth:urgent1}
            \State Create $\ControllerLoc_{i}$
                \Comment{New location for the instruction}
                \label{alg:line:synth:instrq}
            \State Add transition from $\ControllerLoc_{\SymbState_1.\loc}$ to
                $\ControllerLoc_{i}$ with guard $\SymbState_1.\zone$, no clock reset
                \label{alg:line:synth:instrt}

            \If{ $\decision \neq \Wait$}
                \Comment{The decision is a controllable transition}
                \State Match $\decision$ with $(\loc, \_, a, \_, \loc')$
                \State Set $\ControllerLoc_{i}$ invariant to $\TempPred{\SymbState_2}$
                \label{alg:line:synth:instrinv}
                \State \textbf{if} $\ControllerLoc_{\loc'}$ does not exist
                    \textbf{then} create urgent $\ControllerLoc_{\loc'}$ 
                \State Add transition from $\ControllerLoc_{i}$ to $\ControllerLoc_
                {\loc'}$ with guard $\SymbState_2.\zone$, no clock reset, label
                $a$.
                \label{alg:line:synth:instract}
                \If{ $\SymbState_1 \subseteq \SymbState_2$ and $\lnot HasUpperTempCons
                (\SymbState_2.\zone)$}
                \label{alg:line:synth:noup}
                \State Set $\ControllerLoc_{i}$ as urgent
                \label{alg:line:synth:noupurgent}
            \EndIf
            \EndIf

            \For{$\trans_u = (\SymbState_1.\loc, \_, a, \_, \loc')$ uncontrollable
            transition from $\SymbState_1.\loc$}
                \label{alg:line:synth:uncont}
                \State \textbf{if} $\ControllerLoc_{\loc'}$ does not exist
                    \textbf{then} create urgent $\ControllerLoc_{\loc'}$
                \State Add transition from $\ControllerLoc_{i}$ to $\ControllerLoc_
                {\loc'}$ with guard $\top$, no clock reset, label $a$.
                \label{alg:line:synth:uncontt}
            \EndFor
    \EndFor
    \State \Return the controller
    \EndFunction
    \end{algorithmic}
  \end{algorithm}

  \begin{example}
  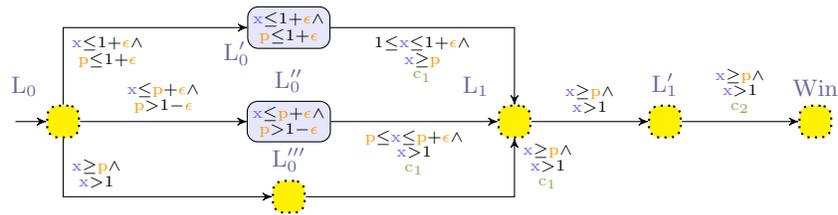
\begin{figure}[b]
    \vspace{-1.5em}
    \centering
    \begin{tikzpicture}[>=stealth', node distance=3cm,initial text=]
        \node[location,urgent,initial left,label={above left:\styleloc{L_0}}] (l0) {};
        \node[location,label={above:\styleloc{L_0''}},right of=l0,
            node distance=3cm] (l02) {
                $\substack{
                    \styleclock{x} \leq \styleparam{p} + \styleparam{\epsilon} \land \\
                 \styleparam{p} > 1 - \styleparam{\epsilon}\\
                }$};
        \node[location,label={[label distance=-.15cm]190:\styleloc{L_0'}},above of=l02,
            node distance=1.25cm] (l01) {$\substack{
                \styleclock{x}\leq 1+\styleparam{\epsilon}  \land \\
                \styleparam{p}\leq 1+\styleparam{\epsilon} \\
                }$};
        \node[location,urgent, label={above:\styleloc{L_0'''}},below of=l02,
            node distance=1cm] (l03) {};

        \node[location, urgent, label={above left:\styleloc{L_1}}, right of=l02, 
            node distance =3cm] (l1) {};

        \node[location, urgent, label={above:\styleloc{L_1'}}, right of=l1, 
            node distance =2cm] (l10) {};

        \node[location, urgent, label={above:\styleloc{Win}}, right of=l10, 
        node distance =2cm] (win) {};

        \draw[->] (l0) |- node[below right] {$\substack{
                \styleclock{x}\leq 1+\styleparam{\epsilon}  \land \\
                \styleparam{p}\leq 1+\styleparam{\epsilon} \\
                }$} (l01);

        \draw[->] (l0) -- node[above] {
            $\substack{
                \styleclock{x} \leq \styleparam{p} + \styleparam{\epsilon} \land \\
             \styleparam{p} > 1 - \styleparam{\epsilon}\\
            }$} (l02);

        \draw[->] (l0) |- node[above right] {            
            $\substack{
                \styleclock{x} \geq \styleparam{p} \land \\
            \styleclock{x} > 1
            }$} (l03);

        \draw[->] (l01) -| node[below, near start] {
            $\substack{
                1 \leq \styleclock{x} \leq 1 + \styleparam{\epsilon} \land \\
                \styleclock{x}\geq \styleparam{p}\\
                \styleact{c_1}
                }$
        } (l1);

        \draw[->] (l02) -- node[below] {
            $\substack{
                \styleparam{p} \leq \styleclock{x} \leq \styleparam{p} + \styleparam{\epsilon} \land \\
                \styleclock{x} > 1\\
                \styleact{c_1}
                }$
        } (l1);

        \draw[->] (l03) -| node[above right] {
            $\substack{
                \styleclock{x} \geq \styleparam{p} \land \\    
                \styleclock{x} > 1\\
                \styleact{c_1}
                }$
        } (l1);

        \draw[->] (l1) -- node[above] {
            $\substack{
                \styleclock{x} \geq \styleparam{p} \land \\        
                \styleclock{x} > 1
                }$
        } (l10);

        \draw[->] (l10) -- node [above] {
            $\substack{
                \styleclock{x} \geq \styleparam{p} \land \\     
                \styleclock{x} > 1\\
                \styleact{c_2}
            }$
        } (win);
    \end{tikzpicture}
    \caption{Controller generated for~\cref{fig:generation_ex_1}. 
    Yellow circles indicate urgent locations.
    \label{fig:generation_ex_2}}
\end{figure}

  The most interesting aspect of the game in~\cref{fig:generation_ex_1} happens in
  the controller generation. Notice that in the strategy we can wait as long as we
  want in $\styleloc{L_0}$ and $\styleloc{L_1}$. Thus, a controller must at some point
  force us to move, which is where the epsilon parameter operates. See~\cref{fig:generation_ex_2}
  for the generated controller. 
  
  From $\styleloc{L_0}$ there are three choices generated from the corresponding strategy entry in~\cref{tab:strategy} due to the pre-processing step described by~\cref{alg:preprocess}. The transition to $\styleloc{L_0'}$ can only be taken if $\styleclock{x} \geq 1$ happens after $\styleclock{x} \geq \styleparam{p}$. The opposite is true for $\styleloc{L_0''}$. $\styleloc{L_0'''}$ is added by the pre-processing in case we arrive in $\styleloc{L_0}$ after we have passed the two boundaries in which case we can continue immediately. Hence, $\styleloc{L_0'''}$ is also made urgent.
  The guards of the transitions from $\styleloc{L_0'}$ and $\styleloc{L_0''}$ along with their invariants ensure that the action $\styleact{c_1}$ is fired at the appropriate time according to the strategy. After arriving in the mirror location $\styleloc{L_1}$, no pre-processing is actually necessary because of a small optimization in the implementation: for a strategy entry $(\SymbState_1,\SymbState_2, \decision)$ if $\SymbState_1\subseteq \SymbState_2$ then it is left untouched. Since we can take the action and stay winning, we do not need to worry about finding bounds. 
\end{example}

\subsection{Properties of the controller}

The controller $\cont$ generated by \Cref{alg:ContSynt} with a strategy
specification $\PartList$ as input satisfies the following properties, as long as
all states of a history $\hist$ match an instruction of $\PartList$ and for $\styleparam{\epsilon} > 0$:
$\cont$ does not cause the run to end prematurely;
$\cont$ does not prevent uncontrollable transitions from happening;
$\cont$ restricts controllable discrete transitions to match the strategy $\PartList$.
Thus, as proved in \cref{app:proofs5}:

\begin{restatable}[Controller Correctness]{theorem}{correctnessController}
    Let $\game$ be a PTG with a reachability objective. Let $\PartList$ be the strategy
 generated by \Cref{alg:main}. Let $\cont$ be the controller synthesized by \Cref{alg:ContSynt} from $\PartList$. Then all runs of 
    $\ParallelComp{\cont}{\game}$ starting in an initial state with a winning parameter valuation and $\styleparam{\epsilon} > 0$ eventually reach a target location. 
\end{restatable}

\section{Implementation and Experimental Evaluation}
\label{sec:expe}
The algorithm proposed by~\cite{OTF-TG} was extended with parameters and implemented
by \cite{dahlsenjensen2024ontheflyalgorithmreachabilityparametric} along with some optimizations. The implementation is an extension to \imitator model checker~\cite{DBLP:conf/cav/Andre21}, supporting a wide set of PTA synthesis algorithms. 

\imitator allows the user to specify a model which consists of parameters, clocks,
a network of parametric timed automata and (for PTG) a partitioning of actions into
controllable and uncontrollable. The winning parameters for a PTG can then be synthesized by querying with property \texttt{Win} to select
\texttt{AlgoPTG}. Thus, 
\texttt{property := \#synth Win(state\_predicate)}
synthesizes parameters for a PTG,
where \texttt{state\_predicate} defines winning states. Predicate \texttt{accepting},
captures all states with an accepting location.

We extended the PTG parameter synthesis algorithm (\texttt{AlgoPTG}) of~\cite{dahlsenjensen2024ontheflyalgorithmreachabilityparametric} and the source code 
can be found on GitHub\footnote{
    \url{https://github.com/imitator-model-checker/imitator}, branch: develop%

}. Our extensions are threefold:
\begin{enumerate}[noitemsep, topsep=0pt]
    \item Forced uncontrollable action semantics (\cref{sec:FTS})
    \item  Strategy synthesis (\cref{sec:InstrList})
    \item Controller synthesis (\cref{sec:ContSynth})
\end{enumerate}

\subsection{Verification}\label{sec:verification}
In order to verify that a controller actually works, we can use \imitator again. The
input game (\cref{fig:generation_ex_1}) and the generated controller (\cref{fig:generation_ex_2})
are combined in a single \imitator model with synchronized actions $\styleact{c_1}$
and $\styleact{c_2}$. By querying with $\texttt{AF Accepting}$, we synthesize the parameters
for which all runs of the composition reach an accepting state. If the controller
is indeed correct, this should reproduce the results of the winning parameter synthesis
algorithm (plus 
$\styleloc{\epsilon} > 0$). 

While \imitator has a tailored algorithm for unavoidability synthesis ($\texttt{AF}$), 
this query can also be solved by viewing the parallel composition as a game, in which all actions are uncontrollable. This game is won by the controller if and only if the goal is unavoidable.
We have found that the algorithm proposed in this paper is faster than the tailored AF-synthesis.
This is used in our experiments to confirm the generated controllers are
correct.

\subsection{Experiments}
We evaluate the scalability on the \textit{Production Cell} case study~\cite{ProdCellCaseStudy}. It has two conveyor belts (incoming/ exiting), a press and a robot with arms $A$/$B$ (\cref{fig:productioncell} shows a visualization). The arms are coupled, moving in unison between the incoming conveyor belt, the press and the exiting conveyor belt. The incoming belt transports unprocessed plates to a buffer zone below arm A. Once the preceding plate has been processed at the press, arm A transfers the unprocessed plate to the press, while arm B transfers the processed plate to the exiting belt. 

We model this in \imitator with 1-4 plates. The reachability goal for the $n$ plate experiment is that all $n$ plates make it to the exiting conveyor belt. More precisely, no two plates must be at the buffer zone simultaneously. If this happens the game is immediately lost via a trap location.
There are several good candidates for parameters: the min/max waiting time between
consecutive plates, the min/max time it takes for the robot to rotate and the processing
time of a plate. For our experiments, we have one parameter for the min waiting time,
while all other quantities are constant. The exact timing of plates arriving at the
incoming conveyor belt and the robot arms movement is uncontrollable, although the environment must respect the constants and min waiting time parameter. 

The PTG model is largely inspired by the one in~\cite{dahlsenjensen2024ontheflyalgorithmreachabilityparametric}.
The only key difference is that with the new forced uncontrollable semantics, the
modelling becomes simpler - there is no need for auxiliary controllable actions to
simulate the environment respecting invariants.
See Appendix~\ref{app:prodcell} for the 1-plate \imitator model.

\begin{figure}
    \centering
    \scalebox{0.7}{
    \begin{tikzpicture}
        \tikzstyle{processed}=[preaction={fill, white}, pattern=crosshatch, pattern color=purple!70!blue, rounded corners];
        \tikzstyle{unprocessed}=[preaction={fill, white},pattern=north west lines, pattern color=blue]
    
        \draw[->,very thick] (-1,.75) -- ++ (.75,0);
        \node at (-1.5,.75) {};
        \draw[->,very thick] (-.25,5.25) -- ++ (-.75,0);
        \draw[rounded corners, pattern={Lines[angle=90,distance=13]},pattern color=gray] (0,0) rectangle ++(6,1.5);
        \draw [rounded corners, pattern={Lines[angle=90,distance=13]},pattern color=gray] (0,4.5) rectangle ++(6,1.5);
    
        \draw (6.5, 2.25) rectangle ++ (1.5,1.5);
        \draw[very thick] (5,3) circle (.75);
        \draw[very thick,fill=white] (4.25,.25) rectangle ++ (1.5,1);
    
        \begin{scope}
            \clip (0,0) rectangle ++(6,1.5);
            \draw[unprocessed] (-.5,.5) rectangle ++ (1,.5);
        \end{scope}
        \begin{scope}
            \clip (0,5) rectangle ++(6,1.5);
            \draw[processed] (-.5,5) rectangle ++ (1,.5);
        \end{scope}
        \draw[unprocessed] (3.5,.5) rectangle ++ (1,.5);
        \draw[processed] (7, 2.5) rectangle ++ (.5,1);
        \draw[processed] (3, 5) rectangle ++ (1,.5);

        \draw[fill=white] (4.75,.5) rectangle ++(.5,2.75);
        \draw[fill=white] (4.75,2.75) rectangle ++(2.75,.5);

            \draw[-,dashed] (.5,.5)--(.5,-.5);
            \draw[-,dashed] (3.5,.5)--(3.5,-.5);
            \draw[<->] (.5,-.5)--(3.5,-.5);
            \node at (2,-.75) {$p \dots \textsc{max\_interval}$};
        
            \node (x) at (7.75,4) {$\textsc{min\_rot} \dots \textsc{max\_rot}$};
            \draw[dashed] (4.75,2.75) rectangle ++(.5,2.75);
            \begin{scope}
                \clip (5.25,3.25) rectangle ++(1.25,1.25);
                \draw[dashed] (5,3) circle (1.25);
                \draw[->] (5.30,4.25)--(5.25,4.25);
                \draw[->] (6.25,3.30)--(6.25,3.25);
            \end{scope}

        \node at (5,.75) {$A$};
        \node at (7.25,3) {$B$};

    \end{tikzpicture}}
    \caption{Production Cell case study. \label{fig:productioncell}}
    \end{figure}
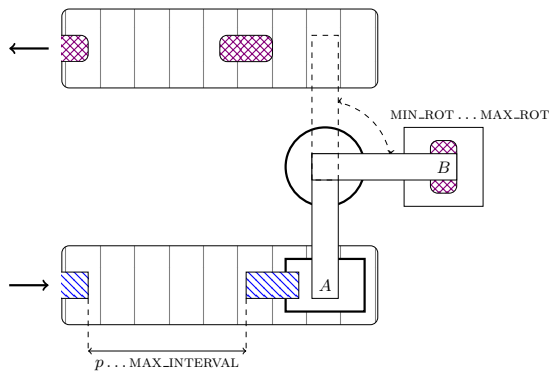

\subsection{Experimental Results}

\begin{table}[!!b]
    \caption{\label{experiment-table}\footnotesize
     Experimental results for 1-4 plates. "P" refers to parameter synthesis (original
     algorithm), while "P+C" also includes
     strategy construction
     with controller synthesis (our contribution).
      Values in parentheses
     indicate the percentage \wrt{}
     the "P" row.
      }

    \footnotesize\medskip\centering
    \begin{tabular}{|c|c||c|c|c|p{.5cm}|c|c|}
            \multicolumn{5}{c}{Solving Game} & \multicolumn{1}{c}{\hspace{1cm}} &
            \multicolumn{2}{c}{Controller Verification}\\
            \cline{1-5} \cline{7-8} 
            Plates & Version & P. Synth. & C. Synth. & Size & & P. Synth (\%) & Size (\%) \\
            \hhline{|=|=#=|=|=|~|=|=|}
            \multirow{2}{*}{1} & P+C & 0.080s & 0.029s & 73 & &\multirow{2}{*}{0.074s (158\%)} &
            \multirow{2}{*}{95 (130\%)} \\
            \cline{2-5} 
            & P & 0.047s& \cellcolor{gray} & 73 & & &  \\
            \cline{1-5} \cline{7-8}
            \multirow{2}{*}{2} & P+C & 0.68s & 0.20s & 496 & & \multirow{2}{*}{0.49s (127\%)} &
            \multirow{2}{*}{378 (76\%)} \\
            \cline{2-5} 
            & P & 0.39s& \cellcolor{gray} & 496 & & & \\
            \cline{1-5} \cline{7-8}
            \multirow{2}{*}{3} & P+C & 2.46s & 0.34s & 1167 & & \multirow{2}{*}{0.46s (26\%)} &
            \multirow{2}{*}{241 (21\%)} \\
            \cline{2-5} 
            & P & 1.76s& \cellcolor{gray} & 1167 & & & \\
            \cline{1-5} \cline{7-8}
            \multirow{2}{*}{4} & P+C & 149s & 2.89s & 7477 & & \multirow{2}{*}{3.33s (5.6\%)} &
            \multirow{2}{*}{703 (9.4\%)} \\
            \cline{2-5} 
            & P & 60s& \cellcolor{gray} & 7477 & & & \\
            \cline{1-5} \cline{7-8}
        \end{tabular}
\end{table}
Results are in \cref{experiment-table} and
experimental details in \cref{app:exp_setup}.
Strategy construction has a 250\% slowdown in the 4-plate experiment, revealing an overhead in strategy construction.
Controller synthesis is very fast, indicating that the majority of the computational
effort occurs during the main algorithm for strategy construction, while controller translation requires minimal additional resources.

Verification speeds up relatively as the model grows. Due to the way the controller is constructed, it always has more discrete locations than the input model, but when exploring the parallel composition, the symbolic state space quickly becomes smaller. This is expected, as the controller restricts the availability of controllable actions, directing the exploration more efficiently.

Synthesis times are 4000\% smaller than those in~\cite{dahlsenjensen2024ontheflyalgorithmreachabilityparametric} (P column). Beyond hardware/library/implementation differences, 
the key conceptual difference is that forced uncontrollable
transitions make the model smaller and avoid re-computations of winning moves.

\section{Conclusion}
\label{sec:concl}
This paper proposes a method to compute a controller for a parametric timed reachability game. The controller is represented by a parametric timed automaton, to be synchronized with the game.
To this end, we introduced a new notion of strategy specifications, which allow for planned actions to happen in a specified time interval. We also updated the semantics of timed games, to properly take care of forced environment transitions.
We demonstrate experimentally that controller synthesis is feasible. Most extra time is spent on generating the strategy specification while solving the game on-the-fly. Generating the controller from the strategy is fast. The experiment also demonstrates that the composition of the game with the controller is guaranteed to reach the goal.

\paragraph{Data Availability.} The models, scripts, and tools to reproduce our experimental evaluation are archived and publicly available at~\cite{dahlsen_jensen_2025_15284734}. 

\bibliographystyle{plain}
\bibliography{bibliography.bib}

\newpage
\appendix

\section*{Appendix}
\section{Proof for Section~\ref{sec:InstrList}: Computed Strategy is Winning}
\label{app:proofs3}

\strategyWinning*
\begin{proof}
    Recall that a state matches an instruction $(\SymbState, \ldots)$ if it is an element of $\SymbState$.
    We will number the instructions of the synthesized strategy specification in the order in which they were added to the strategy specification.
    We will first prove the following invariants:

\medskip
\noindent\textbf{All states in $\Win$ match some instruction in the strategy specification.}

Indeed, states are added to $\Win$ either at the end of an update or when discovering a target location. During an \emph{update}, we add $\NewWin$ to $\Win$ on \cref{alg:line:update:addnewwin} if it contains new winning states. $\NewWin$ is initially empty and accumulates states from $\NewWin_i$ that are added to the strategy simultaneously 
(\crefrange{alg:line:update:pair1a}{alg:line:update:pair1b} and \crefrange{alg:line:update:pair2a}{alg:line:update:pair2b}).
During an \emph{exploration}, if a symbolic state $\SymbState$ matches the target location, we add it to $\Win$, as well as to the strategy specification (\crefrange{alg:line:explore:winupdate}{alg:line:explore:instrlist}).

\medskip

\noindent\textbf{All states match at most one instruction of the strategy specification.}

Before we add an instruction to the list, we remove from $\NewWin_i$ the states that match any other instruction of the strategy specification
(\crefrange{alg:line:update:remove_overlapsA}{alg:line:update:remove_overlapsB}). Therefore, a state can only ever match one instruction of the strategy specification.

\medskip
    
Using the second invariant, we can extend the numbering to states. For all states $\state$, let $\ind(\state)$ denote the index of the instruction $\instr \in \PartList$ matching $\state$, if it exists, or $\ind(\state) = \infty$ otherwise.

We now proceed to prove the property.
Let $\state_0$ be a winning initial state found by \cref{alg:main} and let $\run =
\state_0 \state_1 \ldots$ be a run coherent with the controller strategy associated
with the strategy specification $\PartList$ generated by \cref{alg:main}.
We will prove by induction that $\forall i \in \mathbb{N}$, if $\state_i$ exists in $\run$, and $\ind(\state_i)$ is finite, then either $s_i$ matches the target location, or $s_{i+1}$ exists in $\run$ and $\ind(s_i) > \ind(s_{i+1})$.
    
\textbf{Basis:} Since $\state_0$ is found winning, it enters $Win$, and therefore it matches an instruction of $\PartList$ and $\ind(\state_i)$ is finite.

\textbf{Step:} Let $\state_i$ be a state of the run such that $\ind(\state_i)$ is finite and $\state_{i}$ doesn't match a target location.
Then
$\state_i$ matches an instruction $\instr$ from the strategy specification $\PartList$, by the first invariant. 
Since $\state_i$ does not match a target location, $\instr$ is an instruction of the
form $(\SymbState, \SymbState', \trans_c)$ added during an update. Let $\SymbState_{up}$ be the symbolic state of the zone graph being updated.
Let $\Win_t$ be the winning states gathered by the algorithm at the beginning of the update where $\instr$ was added to $\PartList$. 
We define the uncontrollable states computed at the beginning of the update as
\[\Uncontrollable_t :=\hspace{-1em}\underset{\{(\SymbState', t_u) \mid
    \trans_u \in \TransSet_u \,\land\,\SymbState' =\TempSucc{Succ(\trans_u, \SymbState)}\}}
    {\bigcup} \hspace{-1em}
          Pred (t_u,\SymbState' \setminus \mathit{Win_t}[\SymbState'])\] 
    We have that $\SymbState'$ is a convex zone, $\SymbState' \subseteq Pred(\trans_c, Win_t) \setminus \Uncontrollable_t$ and $\SymbState \subseteq \SafePred(\SymbState', \Uncontrollable_t)$.
    Since $\state_i$ doesn't match a waiting instruction, the decision to wait indefinitely isn't available to the controller. So a discrete transition will be taken and $\state_{i+1}$ exists.
    We distinguish if the next transition that was taken is controllable or uncontrollable.
    
    If a controllable action is taken next, the action is $\trans_c$, taken from $\SymbState' \subseteq Pred(\trans_c, Win_t)$. Therefore, $\state_{i+1} \in Win_t$. From there, $\state_{i+1}$ was found winning before the current update and matches an instruction added to $\PartList$ before the update. As a result, $\ind(\state_{i+1}) < \ind(\state_{i})$. 
    
    Else, an uncontrollable action is taken next. 
    Let $(\delta_u, \trans_u) \in (\mathbb{R}_{\geq 0} \times \TransSet_u)$ such that $\state_i \to^{\delta_u} \state_{\delta_u} \to^{\trans_u} \state_{i+1}$. Define
    \begin{align*}
        \delay_{\mathit{inf}} &= \inf \{ \delay \in \mathbb{R}_{\geq 0} \; | \; \exists
        \state_ {\delay} \in \SymbState' \; \state_i \to^{\delay} \state_{\delay} \}\\
        \delay_{\mathit{sup}} &= \sup \{ \delay \in \mathbb{R}_{\geq 0} \; | \; \exists
        \state_{\delay} \in \SymbState' \; \state_i \to^{\delay} \state_{\delay} \}\enspace.
    \end{align*}
    
We must have $\delay_u \leq \delay_{\mathit{sup}}$, else $\trans_c$ would have been
taken. \jvdp{I don't get this. Why is some $\trans_c$ enabled here?}
Next, we show that $\state_{\delay_u} \notin \Uncontrollable_t$
by distinguishing two cases:
\begin{itemize}
\item    If $\state_{\delay_u} \in \SymbState'$, note that $\SymbState' \subseteq Pred(\trans_c, Win_t) \setminus \Uncontrollable_t$, so indeed $\state_{\delay_u} \notin \Uncontrollable_t$.
\item    
    Else, $\state_{\delay_u} \notin \SymbState'$. Then $\delay_u \leq \delay_{\mathit{inf}}$
    because $\SymbState'$ is convex. We cannot reach $\SymbState'$ from $\state_i$
    without going through $\state_{\delay_u}$. But $\state_i \in \SymbState \subseteq \SafePred (\SymbState', \Uncontrollable_t)$. Therefore, $\state_{\delay_u} \notin \Uncontrollable_t$.
\end{itemize}

Since $\Uncontrollable_t \supseteq Pred (\trans_u, \SymbState_{\trans_u} \setminus
    \Win_t[\SymbState_{\trans_u}])$, we have $\state_{i+1} \in \Win_t$. From there, $\state_{i+1}$ was found winning before the current update, and matches an instruction added to $\PartList$ before the update. As a result, $\ind(\state_{i+1}) < \ind(\state_{i})$. 
    
    So, following the strategy specification from any state in $\Win$ will result in a run with a strictly decreasing $\ind(\state_i)$, that does
    not terminate until the target is reached. Since $\ind(\state_i)$ takes values in the positive integers, a target location is always reached.
\qed
\end{proof}

\section{Proofs for Section~\ref{sec:FTS}: Correctness Forced Transitions}

\label{app:proofs4}

\subsection{Correctness of upper/lower Temporal Closure}
\begin{lemma}
    The upper (resp. lower) temporal closure $\UpTempClos{\zone}$ (resp. $\LoTempClos{\zone}$) of a zone $\zone$ is the set of valuations $\val$ such that there exists
    a sequence of valuations $(\val_n)_{n\in \mathbb{N}}$ in $\zone$ such that:
    For all $n \in \mathbb{N}$, there exists a \emph{strictly} positive delay $\delay_n > 0$ such that $\val_n \to^{\delay_n} \val_{n+1}$ (resp. $\val_{n+1} \to^{\delay_n} \val_n$) and $\val$ is the limit of the sequence $(\val_n)_{n\in \mathbb{N}}$.

\end{lemma}

\begin{proof}
    Let $\zone$ be a zone represented by the formula $\ZoneFormula$.
    Let $\val$ be a valuation of $\UpTempClos{\zone}$, the upper temporal closure of $\zone$. Then we have
    for all upper temporal constraints $\UpTempCons \in \UpTempConsSet(\zone)$ of the form $\styleclock{x} \sim \ParamLinearTerm$, that $\val$ satisfies the constraint $\styleclock{x} \leq \ParamLinearTerm$.
    Similarly,
    for all lower temporal constraints $\LoTempCons \in \LoTempConsSet(\zone)$ of the form $\styleclock{x} \sim \ParamLinearTerm$, we have that $\val$ satisfies the constraint $\styleclock{x} > \ParamLinearTerm$.
    For all diagonal constraints $\phi$, $\val$ satisfies $\phi$.

    We prove the two directions of the property.
    For the first direction, define $$\delay_{\mathit{sup}} = \underset{(\styleclock{x}
    \sim \ParamLinearTerm) \in \LoTempConsSet}{\max} \val(\styleclock{x}) - \val(\ParamLinearTerm)\enspace.$$
    Since $\zone$ has a finite number of lower temporal constraints and for each temporal
    constraint, $\val(\styleclock{x}) - \val(\ParamLinearTerm) > 0$, we have that
    $\delay_{\mathit{sup}} > 0$.

    We will now prove for all $n \in \mathbb{N}$, that $\val_n = \val - 2^{-(n+1)}
    \delay_{\mathit{sup}}$. So let $n \in \mathbb{N}$. We claim that $\val_n$ satisfies
    all constraints of $\zone$, by distinguishing three types of constraints:
\begin{itemize}
\item
    Note that $\val_n$ satisfies all diagonal constraints of $\zone$ because $\val$ does, and diagonal constraints are invariant under time progress.
\item
    For all lower temporal constraints $\styleclock{x} \sim \ParamLinearTerm$ of $\zone$, we have  $\val_n(\styleclock{x}) = \val(\styleclock{x}) - 2^{-(n+1)} \delay_{sup} > \val(\styleclock{x}) - \delay_{sup} \geq \val(\ParamLinearTerm) = \val_n(\ParamLinearTerm) $.
    Therefore, $\val_n$ satisfies $\styleclock{x} \sim \ParamLinearTerm$.
\item
    For all upper temporal constraints $\styleclock{x} \sim \ParamLinearTerm$ of $\zone$, we have $\val_n(\styleclock{x}) < \val(\styleclock{x}) \leq \val(\ParamLinearTerm) = \val_n(\ParamLinearTerm)$.
    Therefore, $\val_n$ satisfies $\styleclock{x} \sim \ParamLinearTerm$.
\end{itemize} 
    Thus, for all $n \in \mathbb{N}$, $\val_n \in \zone$, $\val_{n+1} = \val_n + 2^
    {-(n+1)} \delay_{\mathit{sup}}$ with $2^{-(n+1)} \delay_{\mathit{sup}} > 0$. Finally,
    $\val$ is the limit of the sequence $(\val_n)_{n\in \mathbb{N}}$ when $n$ goes to infinity.
\medskip

    Conversely, let $\val$ be a valuation such that there exists a sequence of valuations $(\val_n)_{n\in \mathbb{N}}$ with limit $\val$ and for all $n \in \mathbb{N}$, there exists a delay $\delay_n > 0$ such that $\val_n \to^{\delay_n} \val_{n+1}$.

    Then we have $\val = \val_0 + \delay_{\mathit{sup}}$ where $\delay_{\mathit{sup}} = \sum_{k=0}^{\infty} \delay_n > 0$. We show that this valuation $\val$ satisfies
    all constraints in $\zone$, by distinguishing three types of constraints:
\begin{itemize}
\item
    $\val$ satisfies all diagonal constraints of $\zone$ because $\val_0$ does, and diagonal constraints are invariant under time progress.
\item
    For all lower temporal constraints $\styleclock{x} \sim \ParamLinearTerm$ of $\zone$,
    we have $\val(\styleclock{x}) = \val_0(\styleclock{x}) + \delay_{\mathit{sup}}
    > \val_0 (\styleclock{x}) \geq \val_0(\ParamLinearTerm) = \val(\ParamLinearTerm)$.
    Therefore, $\val$ satisfies $\styleclock{x} > \ParamLinearTerm$.
\item
    For all upper temporal constraints $\styleclock{x} \sim \ParamLinearTerm$ with ${\sim} \in \{ < ; \leq\}$ of $\zone$, for all $\epsilon > 0$, there exists $n \in \mathbb{N}$ such that $\val(\styleclock{x}) - \epsilon < \val_n(\styleclock{x}) \sim \val_n(\ParamLinearTerm) = \val(\ParamLinearTerm)$.
    So $\val(\styleclock{x}) \leq \val(\ParamLinearTerm)$, and $\val$ satisfies $\styleclock{x} \leq \ParamLinearTerm$.
\end{itemize}

\noindent Thus, $\val$ belongs to the temporal upper closure $\UpTempClos{\zone}$ of $\zone$.\qed
\end{proof}

\subsection{Theorem~\ref{thmt@@temporalBoundZone}:
Correctness of Temporal upper/lower Bound}
\temporalBoundZone*
\begin{proof}
    Let $\zone$ be a zone.

    First, assume that $\zone$ has no upper temporal constraint. Then the temporal upper bound of $\zone$ is the empty union, so
    $\UpTempGlobBound{\zone} = \emptyset$.  
    We claim that for all valuations $\val$ in $\zone$ and delays $\delay \in \mathbb{R}_{\geq 0}$, we have $\val+\delay$
    satisfies all constraints in~$\zone$:

\begin{itemize}
\item For all diagonal constraints $\phi$ of $\zone$, $\val + \delay $ satisfies $\phi$ because $\val$ does, and diagonal constraints are invariant under time progress.
\item
    Let $\styleclock{x} \sim \ParamLinearTerm$ be any other temporal constraint of $\zone$. Then ${\sim} \in \{\geq,>\}$, since $\zone$ has no upper temporal constraints. $(\val+ \delay)(\styleclock{x}) \geq \val(\styleclock{x}) \sim \val(\ParamLinearTerm) = (\val + \delay) (\ParamLinearTerm)$. Therefore, $\val+ \delay$ satisfies $\styleclock{x} \sim \ParamLinearTerm$.
\end{itemize}
    Thus, for all valuations $\val$ in $\zone$ and for all delays $\delay \in \mathbb{R}_{\geq 0}$, $\val + \delay $ is in $\zone$.  We have:
    $\UpTempGlobBound{\zone} = \{ \val + \delay_{\mathit{sup}} \mid \val \in \zone \text{
        such that } \delay_{\mathit{sup}} =  \sup \{ \delay \in \mathbb{R}_{\geq 0} \mid \val + \delay \in \zone \} \in \mathbb{R}_{\geq 0} \} = \emptyset$.

    \medskip
    
    Next, we consider the case that $\zone$ has at least one upper temporal constraint. Let $v$ be an arbitrary valuation in $\zone$,
    and let $\UpTempCons \in \UpTempConsSet(\zone)$ 
    be an arbitrary upper temporal constraint of $\zone$ of the form $\styleclock{x} \sim \ParamLinearTerm$. Let $\delay_{\UpTempCons} = \val(\ParamLinearTerm) - \val(\styleclock{x})$. Since $\val$ is in $\zone$, it satisfies $\UpTempCons$ and therefore $\delay_{\UpTempCons} \in \mathbb{R}_{\geq 0}$.
    \medskip

    Let $\delay_{\mathit{sup}} = \underset{\UpTempCons \in \UpTempConsSet(\zone)}
    {\min} \delay_{\UpTempCons}$ and let $\Phi_{\delay_{\mathit{sup}}} \subseteq \UpTempConsSet$
    be the set of temporal upper temporal constrains $\UpTempCons$ of $\zone$ such
    that $\delay_{\mathit{sup}} = \delay_{\UpTempCons}$.
    We claim that for all $0 \leq \delay < \delay_{sup}$, 
    the valuation $\val + \delay$ is in $\zone$:
    \begin{itemize}
\item    $\val + \delay$ satisfies all upper temporal constraints $\UpTempCons$ of
$\zone$ because, for all $\UpTempCons \in \UpTempConsSet(\zone)$, $\delay < \delay_
{\mathit{sup}} \leq \delay_{\UpTempCons}$.
    
\item    $\val + \delay$ satisfies all diagonal constraints of $\zone$ because $\val$ does, and those constraints are invariant under time progress.
    
\item    $\val + \delay$ satisfies all lower temporal constraints of $\zone$ because for all clock $\styleclock{x}$, $(\val+\delay)(\styleclock{x}) \geq \val(\styleclock{x})$. 
    \end{itemize}
    
    \medskip

    Furthermore, for all $\delay > \delay_{\mathit{sup}}$, $\val + \delay$ invalidates
    the constraints in $\Phi_{\delay_{\mathit{sup}}}$ and therefore, $\val + \delay$
    is \emph{not} in $\zone$.
    We thus have that $\delay_{sup} =  \sup \{ \delay \in \mathbb{R}_{\geq 0} \mid \val + \delay \in \zone \}$.

    \medskip
    We will now prove the $\supseteq$-direction of the Property,
    distinguishing two cases.

    First, consider the case that there exists a strict constraint $\UpTempCons \in \Phi_{\delay_{\mathit{sup}}}$. Then it is of the form $ \styleclock{x} < \ParamLinearTerm$ with $\delay_{\mathit{sup}} = \delay_{\UpTempCons}$ and $\val + \delay_{\mathit{sup}}$ satisfies $\styleclock{x} = \ParamLinearTerm$.
    Now $\delay_{\mathit{sup}} > 0$, since otherwise $\val = \val + \delay_{\mathit{sup}}$ and $\val$ satisfies both $ \styleclock{x} < \ParamLinearTerm$ and $\styleclock{x} = \ParamLinearTerm$ (contradiction).

    \begin{itemize}
    \item
    For all upper temporal constraints $\UpTempCons$ of the form $ \styleclock{x} \sim \ParamLinearTerm$ of $\zone$, $(\val + \delay_{\mathit{sup}}) (\styleclock{x}) \leq \val(\styleclock{x}) + \delay_{\UpTempCons} = \val(\ParamLinearTerm)$.
\item
    For all diagonal temporal constraints$\phi$ of $\zone$, $\val+ \delay_{\mathit{sup}}$ satisfies $\phi$ because $\val$ does, and diagonal constraints are invariant under time progress.
\item
    For all lower temporal constraints $\LoTempCons$ of the form $\styleclock{x} \sim \ParamLinearTerm$ of $\zone$, $(\val+\delay_{\mathit{sup}})(\styleclock{x}) > \val(\styleclock{x}) \sim \val(\ParamLinearTerm) = (\val+\delay_{\mathit{sup}})(\ParamLinearTerm)$.
\end{itemize}

    Therefore, $\val+\delay_{\mathit{sup}}$ belongs to the intersection of the zone $\styleclock{x} = \ParamLinearTerm$ and the temporal upper closure of $\zone$. It is in the temporal upper bounds associated with $\UpTempCons$ and thus, in the temporal upper bound of $\zone$.
    \medskip

    Next, consider the case that all constraints in $\Phi_{\delay_{\mathit{sup}}}$
    are non-strict. We claim that $(\val + \delay_{\mathit{sup}})$ satisfies all constraints
    $\phi$ in $\zone$:
    \begin{itemize}
\item
    All constraints $\UpTempCons \in \Phi_{\delay_{\mathit{sup}}}$ such that $\UpTempCons$ is of the form $ \styleclock{x} \leq \ParamLinearTerm$, $(\val + \delay_{\mathit{sup}})$ satisfy $\styleclock{x} = \ParamLinearTerm$, and thus $(\val + \delay_{\mathit{sup}})$ satisfies $\UpTempCons$.
\item
    For all other upper temporal constraints $\UpTempCons$ in $\zone$ of the form $\styleclock{x} \sim \ParamLinearTerm$, $(\val + \delay_{\mathit{sup}})(\styleclock{x}) \leq \val(\styleclock{x}) + \delay_{\UpTempCons} = \val(\ParamLinearTerm) = (\val+ \delay_{\mathit{sup}})(\ParamLinearTerm)$.
\item
    For all diagonal constraints $\phi$ of $\zone$, $\val+ \delay_{\mathit{sup}}$ satisfies $\phi$ because $\val$ does, and diagonal constraints are invariant under time progress.
\item
    For all lower temporal constraints $\LoTempCons$ of the form $\styleclock{x} \sim \ParamLinearTerm$ of $\zone$, $(\val+\delay_{\mathit{sup}})(\styleclock{x}) > \val(\styleclock{x}) \sim \val(\ParamLinearTerm) = (\val+\delay_{\mathit{sup}})(\ParamLinearTerm)$.
\end{itemize}
    
    As a result, $\val+\delay_{\mathit{sup}}$ is in $\zone$.
    Let $\UpTempCons$ be a constraint of $\Phi_{\delay_{\mathit{sup}}}$, $\val+\delay_{\mathit{sup}}$ belong to the upper temporal bound of $\zone$ associated with $\UpTempCons$. Thus, $\val+\delay_{\mathit{sup}} \in \UpTempGlobBound{\zone}$.

    We can now conclude that the set $\{ \val + \delay_{\mathit{sup}} \mid \val \in \zone \text{
        such that } \delay_{\mathit{sup}} =  \sup \{ \delay \in \mathbb{R}_{\geq 0} \mid \val + \delay \in \zone \} \in \mathbb{R}_{\geq 0} \}$ is included in the union of temporal upper bounds associated to temporal upper constraints of $\zone$.
    \medskip

    Conversely, in order to prove the $\subseteq$-direction of the Property, let $\val \in \UpTempGlobBound{\zone}$ be a valuation in the temporal upper bound of $\zone$. Then there exists $\UpTempCons \in \UpTempConsSet$ of the form $\val(\styleclock{x}) \sim \ParamLinearTerm$ such that $\val \in \UpTempBound{\zone}{\UpTempCons}$.   
    
    If $\UpTempCons$ is non-strict, then $\val$ belongs in $\zone$ and $\val(\styleclock{x}) = \val(\ParamLinearTerm)$. 
    Therefore, for all delays $\delay > 0$, $(\val + \delay) (\styleclock{x}) = \val(\styleclock{x}) + \delay > \val(\ParamLinearTerm) = (val + \delay)(\ParamLinearTerm)$.
    Thus, $\val + \delay$ does not satisfy $\UpTempCons$ and is not in $\zone$. 
    We have $ \sup \{ \delay \in \mathbb{R}_{\geq 0} | \val + \delay \in \zone \} = 0$ and $\val = \val + 0$.

    \medskip

    If $\UpTempCons$ is strict, then $\val$ belongs in the temporal upper closure of $\zone$ and $\val(\styleclock{x}) = \val(\ParamLinearTerm)$.
    There exists a sequence $(\val_n)_{n \in \mathbb{N}}$
    of valuations $\val_n$ of $\zone$, such that for all $n \in \mathbb{N}$, there exists a strictly positive delay $\delay > 0$ such that $\val_n \to^{\delay} \val_{n+1}$ and $\val$ is the limit of $(\val_n)_{n \in \mathbb{N}}$.
    
    For all $n \in \mathbb{N}$, there exists a delay $\delay_n$ such that $\val_0
    \to^{\delay_n} \val_n$. The sequence $(\delay_n)_{n \in \mathbb{N}}$ is strictly
    increasing and has a limit $\delay_{\mathit{sup}} \in \DelaySet$ such that $\val_0
    \to^{\delay_{\mathit{sup}}} \val$.
    For all $\delay < \delay_{\mathit{sup}}$, there exists $n \in \mathbb{N}$ such
    that $\delay < \delay_n < \delay_{\mathit{sup}}$. Since $\val_n$ is in $\zone$,
    we have, by convexity of $\zone$, that $\val + \delay$ belongs to $\zone$ as well.
    
    Furthermore, for all finite delays $\delay\geq\delay_{\mathit{sup}}$, we have $
    (\val_0 + \delay)(\styleclock{x}) = \val_0(\styleclock{x}) + \delay \geq \val_0
    (\styleclock{x}) + \delay_{\mathit{sup}} = (\val_0 + \delay_{\mathit{sup}}) (\ParamLinearTerm) = \val_0(\ParamLinearTerm)$. Therefore, $\val + \delay$ does not satisfy the constraint $\styleclock{x} < plt$ which is the constraint $\UpTempCons$ of $\zone$. Thus, $\val+\delay$ is not in $\zone$.
    
    Thus, we have $\val_0 \in \zone$ and $\delay_{\mathit{sup}} = \sup \{ \delay \in \mathbb{R}_{\geq 0} | \val_0 + \delay \in \zone \}$ such that $\val = \val_0 + \delay_{\mathit{sup}}$.

    We can now conclude that $\UpTempGlobBound{\zone} = \underset{\UpTempCons \in \UpTempConsSet}{\bigcup} \UpTempBound{\zone}{\UpTempCons}$.
\qed
\end{proof}

\newpage

\section{Proofs for Section~\ref{sec:ContSynth}: Correctness of the Controller}
\label{app:proofs5}

Let $\run$ be a run of the parallel composition $\ParallelComp{\cont}{\game}$ with a parameter valuation $\val$ such that $\val(\epsilon) > 0$.
\begin{lemma}
For all states $(\loc^{\cont}, \loc^{\game}, \val)$ of $\run$, $\loc^{\cont}$ is the
mirror location of $\loc^{\game}$ or it is an instruction location $\loc^{\instr}$
where $\instr = (\SymbState_1, \_ , \_ ) $ and $\loc^{\game} = \SymbState_1.\loc$.
\end{lemma}

\begin{proof}
    The run starts in the initial game location $\loc_0$ and its mirror in the controller $q_{\loc_0}$. From then on, when a mirror location $q_{\loc}$ is reached in the controller, it is done with a parallel discrete transition with a game transition that leads to some location $\loc$ which is mirrored by $q_{\loc}$. Since there is no internal game transition, the location of the game is unchanged until we leave $q_{\loc}$.
    An instruction location $q_i$ where $\instr = (\SymbState_1, \_ , \_ )$ of the controller
    is reached by an internal transition of the controller from a mirror location $q_{\loc}$ where $\loc = \SymbState_1.\loc$. Since the location of the game matches $\loc$ in $q_{\loc}$, and an internal controller transition does not change the game location, we still have that the location of the game matches $\loc$ when entering $q_{\loc}$. Since there is no internal game transition, the location of the game is unchanged until we leave $q_{\loc}$.
\end{proof}

\begin{lemma}
Run $\run$ invalidates a controller invariant if and only if it ends in a state $\state = (\loc^{\cont}, \loc^{\game}, \val)$ with $\val(\epsilon) = 0$ or when $\loc^{\cont}$ is a mirror location $q_{\loc^{\game}}$ and there is no matching instruction for $(\loc^{\game}, \val)$ in $\PartList$.
\end{lemma}
\begin{proof}

    If $\loc^{\cont}$ is an instruction location $q_i$ with $\instr = (\SymbState_1, \SymbState_2 , \decision )$, the associated instruction after processing by \textsc{AddEpsilonBounds}, we proceed as follows.
    
    If $\decision = Wait$, the invariant of $q_i$ is then $\top$ and cannot be invalidated.
    Else $\decision = \trans_c \in \TransSet_c$.
    If $\SymbState_2$ has an upper temporal constraint (possibly because it is an
    epsilon lower bound obtained from \textsc{AddEpsilonBounds}) then $\SymbState_1 \subseteq \TempPred{\SymbState_2}$. All states of $\SymbState_1$ can reach $\SymbState_2$. Furthermore, the invariant of $q_i$ has been set to $\TempPred{\SymbState_2}$. At the temporal upper bound of $\TempPred{\SymbState_2}$ there is a transition available synchronized with the game transition $\decision$. Therefore, a transition is forced, and the invariant does not block it.

    If $\SymbState_2$ has no upper temporal constraint, then it was obtained by \textsc{AddEpsilonBounds}
    from a symbolic state with no upper temporal constraint. We have $\SymbState_1 \subseteq \SymbState_2$. Then, $\state$ has an enabled transition synchronized with the game transition $\trans_c$ and $q_i$ is urgent, so a transition is forced, and the invariant does not block it.

    If $\loc^{\cont}$ is a mirror location and $\state$ matches an instruction from $\PartList$.
    Then let $\instr = (\SymbState_1, \SymbState_2 , \decision )$ be an instruction matched by $\state$.

    If $\decision= \Wait$ or $\SymbState_2$ has an upper temporal constraint, then
    $\instr$ is not modified by \textsc{AddEpsilonBounds}.
    $\loc^{\cont} = q_{\SymbState_1.\loc}$ and $\val$ satisfies $\SymbState_1.\zone$, therefore some internal transition of the controller from $q_{\loc}$ to $q_i$ is enabled in $\state$. 
    Since $q_{\loc}$ is urgent and by the forced transition semantics, a discrete
    transition is forced to occur from $\state = (\loc^{\cont}, \loc^{\game}, \val)$.

    Otherwise, if $d\neq \Wait$,
    the temporal predecessors $\TempPred{\SymbState_2}$ are covered by $\SymbState_2$ and the temporal predecessors of its epsilon lower bounds for all parameter valuation with $\val(\styleparam{\epsilon}) > 0$.
    In this case, $\SymbState_1 $ is split between $\SymbState_2$ and the temporal predecessors of the epsilon lower temporal bounds of $\SymbState_2$. Since $\SymbState_1 \subseteq \TempPred{\SymbState_2}$, all states $\state$ of the parallel composition $\ParallelComp{\cont}{\game}$ with a strictly positive value for $\styleparam{\epsilon}$ matching an instruction of $\PartList$ before \textsc{AddEpsilonBounds}, still match an instruction of the strategy specification obtained with \textsc{AddEpsilonBounds}.

    Thus, they have an internal controller transition from $\loc^{\cont}$, which can be forced instead of breaking the urgent location.
\end{proof}

\begin{lemma}
    Run $\run$ is coherent with the strategy described by $\PartList$ until it is
    blocked by a controller invariant.
\end{lemma}

\begin{proof}

    Recall that $\run$ is a run in $\ParallelComp{\cont}{\game}$
    with parameter value $\val$. Assume that $\hist$ is a history of $\run$ that is so far coherent with strategy $\PartList$. We
    will show that the next step (a transition or termination)
    is still coherent with $\PartList$. Assume that $\hist$ ended
    in $ls(\hist) = (\loc^{\cont}, \loc^{\game}, \val)$, where
    $\loc^{\cont}$ is the mirror location of $\loc^{\game}$.

    If there is no instruction matching $\state = (\loc^{\game}, \val)$ after processing by \textsc{AddEpsilonBounds}, then there is no transition available from $\loc^{\cont}$ in the controller. The run $\run$ violates the urgent nature of $\loc^{\cont}$ and terminates correctly.
    
    Else, there exists an internal controller transition to an instruction location $q_i$ tied to some instruction $\instr = (\SymbState_1, \SymbState_2 , \decision )$  matching $\state$ after processing by \textsc{AddEpsilonBounds}.
    
    If $\decision = \Wait$, no controllable transition is available from $q_i$. This
    corresponds to a controller decision $\strat(\hist) = \infty$ which is coherent with $\PartList$. The run $\run$ either stays in $q_i$ forever or takes an uncontrollable transition $\trans_u$ after some delay $\delay \leq \infty$.

    Else, $\decision$ is a controllable action $\trans_c$ from $\TransSet_c$.
    The instruction $\instr$ has been obtained from an instruction $\instr' =  (\SymbState_1', \SymbState_2' , \decision' )$ through processing by \textsc{AddEpsilonBounds} either directly if $\SymbState_2'$ has an upper temporal constraint or from splitting $\instr'$. In any case, we have $\SymbState_1 \subseteq \SymbState_1'$ and $\SymbState_2 \subseteq \SymbState_2'$.

    Furthermore, either $q_i$ is urgent or its invariant is $\TempPred{\SymbState_2}$, where $\SymbState_2 \subseteq \Inv(\loc^{\game})$ has an upper temporal constraint. By ending the run $\run$ here, we would invalidate a controller invariant before invalidating a PTG invariant.

    If $\run$ doesn't invalidate a controller invariant, then a discrete transition occurs.
    The next transition cannot be an internal transition of the controller, since there
    are no internal transitions leaving $q_i$. 
    All transitions of the PTG have a label synchronized with the controller, so internal transitions of the PTG cannot occur either.
    Therefore, either a controllable parallel discrete transition or an uncontrollable parallel discrete transition occurs.

    In the case of a controllable parallel discrete transition, there is only one action synchronized with a controllable label available in $q_i$: a transition with the label of $\trans_c$ from $\SymbState_2$. Since no two transitions available in the same state of the PTG share a label, we are ensured that the only controllable transition available in $q_i$ is $\trans_c$.
    We have that from a state of $\SymbState_1 \subseteq \SymbState_1'$, we take the transition $\trans_c$ after a delay $\delay$ leading to $\SymbState_2 \subseteq \SymbState_2'$. This is coherent with the strategy described by the original instruction list.

    Else, an uncontrollable parallel discrete transition is taken next. It is taken from the invariant of $q_i$.
    Either $q_i$ is urgent and the transition is taken immediately or its invariant is $\TempPred{\SymbState_2}$. In any case, the uncontrollable action is taken within a delay smaller or equal to a delay leading to $\SymbState_2 \subseteq \SymbState_2'$. This is coherent with the strategy described by the original instruction list.
    
    Finally, the next parallel discrete transition is coherent with the strategy specification $\PartList$ and leads to a mirror location in the controller.

\end{proof}

\begin{lemma}
    For all states $(\loc^{\cont}, \loc^{\game}, \val)$ of $\run$ such that an uncontrollable transition $\trans_u$ is available in the game, the controller cannot apply timed transition (with strictly positive delay) or parallel discrete transition before enabling $\trans_u$.
\end{lemma}

\begin{proof}
    If $\loc^{\cont}$ is a mirror location, it is urgent and only has internal outgoing transitions leading to an instruction location.
    If $\loc^{\cont}$ is an instruction location $q_i$ with $\instr = (\SymbState_1, \_ , \_ ) $, then  $\loc^{\game} = \SymbState_1.l$. From there, $\trans_u$ is an uncontrollable transition from $\loc^{\game} = \SymbState_1.l$ and thus has a synchronized transition in the controller from $\loc^{\cont} = q_i$ with  guard $\top$. Thus, $\trans_u$ is enabled.
\end{proof}

Using the previous Lemmas, we can now show the correctness of the generated controller:

\correctnessController*

\begin{proof}
    For all winning states $\state$, the strategy specification $\PartList$ generated
    by \Cref{alg:main} has an instruction matched by $\state$. Furthermore, the strategy
    $\strat$ associated with $\PartList$ is winning and keeps us in winning states until the target is reached. Let $\cont$ be the controller generated by \cref{alg:ContSynt} with $\PartList$ as input. All states of the game $\game$ of a run on the parallel composition $\ParallelComp{\cont}{\game}$ have a matching instruction in $\PartList$ until a target location is reached. Since $\strat$ is winning, all runs on $\ParallelComp{\cont}{\game}$ eventually reach a target location.
\end{proof}

\newpage
\section{Experimental Setup}
\label{app:exp_setup}
All experiments were run on a single core of a computer with an Intel Core i5-1135G7 CPU @ 2.40GHz with 16 GB of RAM running Ubuntu 22.04.5 LTS. \label{app:exp_setup}
In order to measure the additional computation time incurred by constructing the winning strategy, two versions of the algorithm are run for each plate: one that purely does parameter synthesis via the original algorithm and our new version that does parameter synthesis and strategy synthesis as well as generates a controller. The generated controller is verified as described in~\cref{sec:verification}, re-using the algorithm with strategy synthesis turned off.

We run the experiments 5 times and report the average time and symbolic state space size. A timeout of 2 hours is used.

\section{\imitator Model of the Production Cell}
An \imitator model of the 1-plate production cell can be seen in~\cref{fig:model_plate,fig:model_robot,fig:model_broadcaster}. The model is composed of three automata: the plate, the robot and the broadcaster. Arrows with same color signify a synchronized action in \imitator.
\label{app:prodcell}
\begin{figure}[h!]
  \centering
  \includegraphics[width=.61\textwidth]{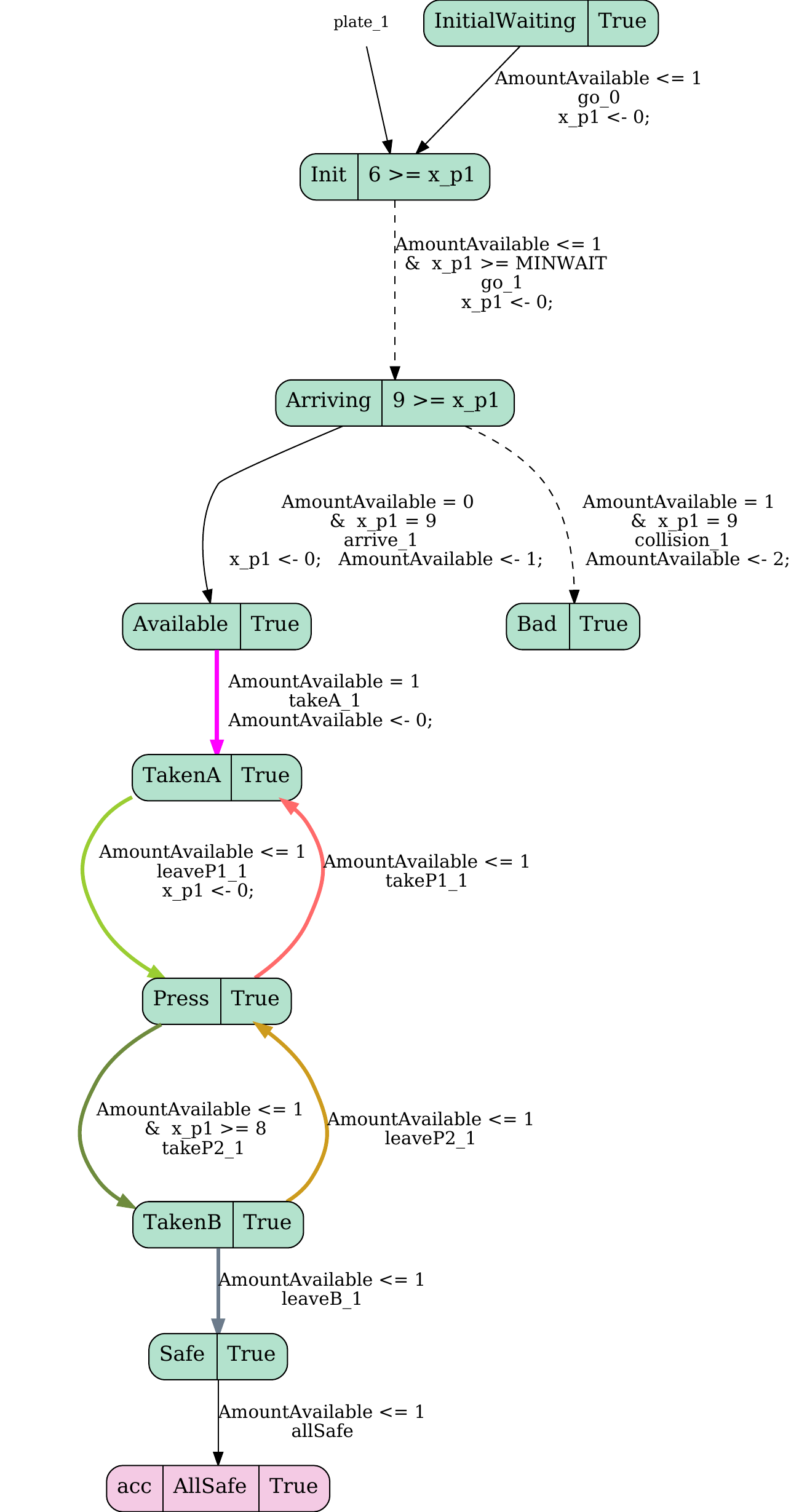}
  \caption{\imitator model for the Production Cell -- plate component}
  \label{fig:model_plate}
\end{figure}

\begin{figure}[h!]
  \hspace{-2.5cm}
  \includegraphics[width=1.2\textwidth]{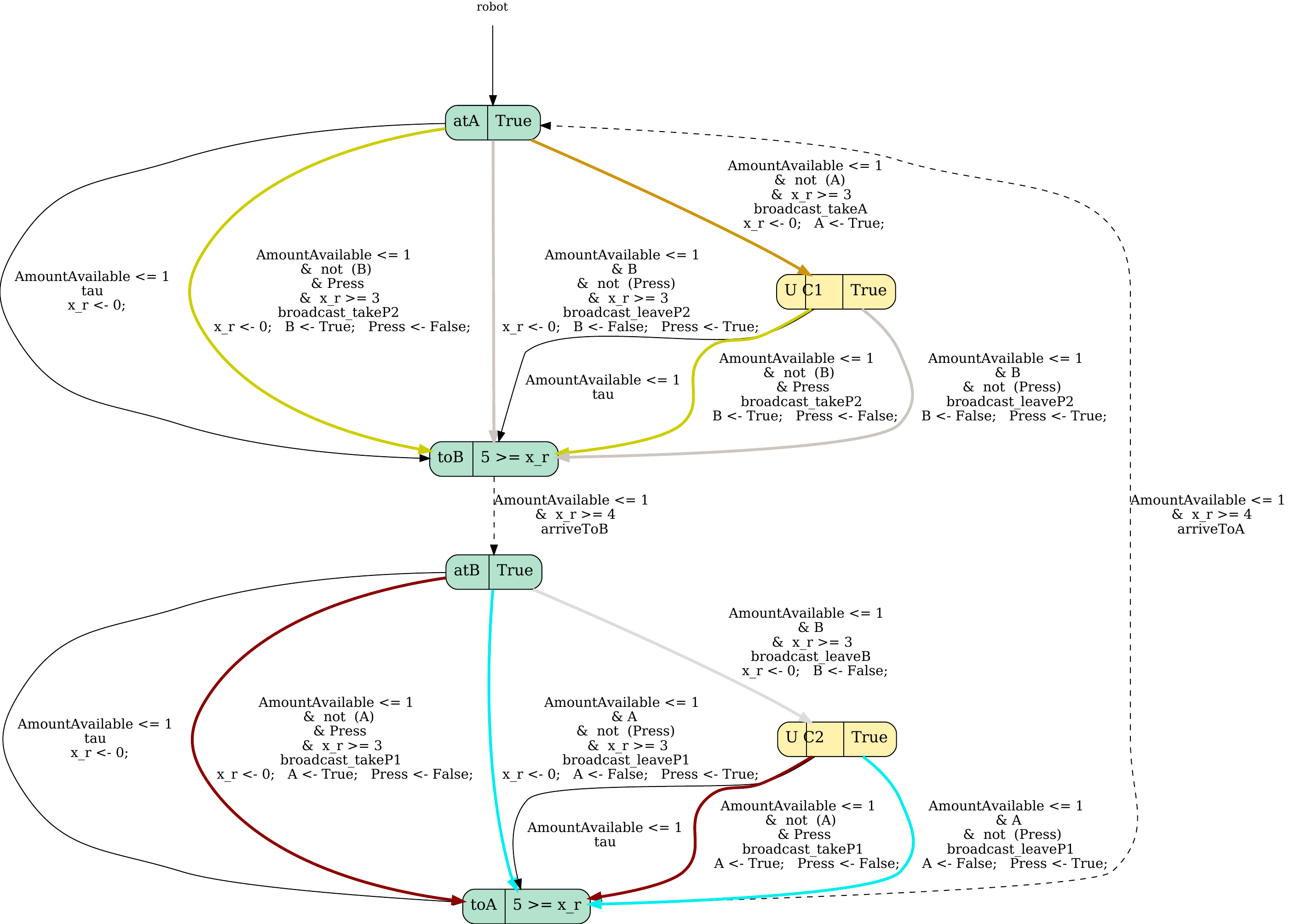}

  \caption{\imitator model for the Production Cell -- robot component}
  \label{fig:model_robot}
\end{figure}

\begin{figure}[h!]
  \hspace{-2.5cm}
  \includegraphics[width=1.5\textwidth]{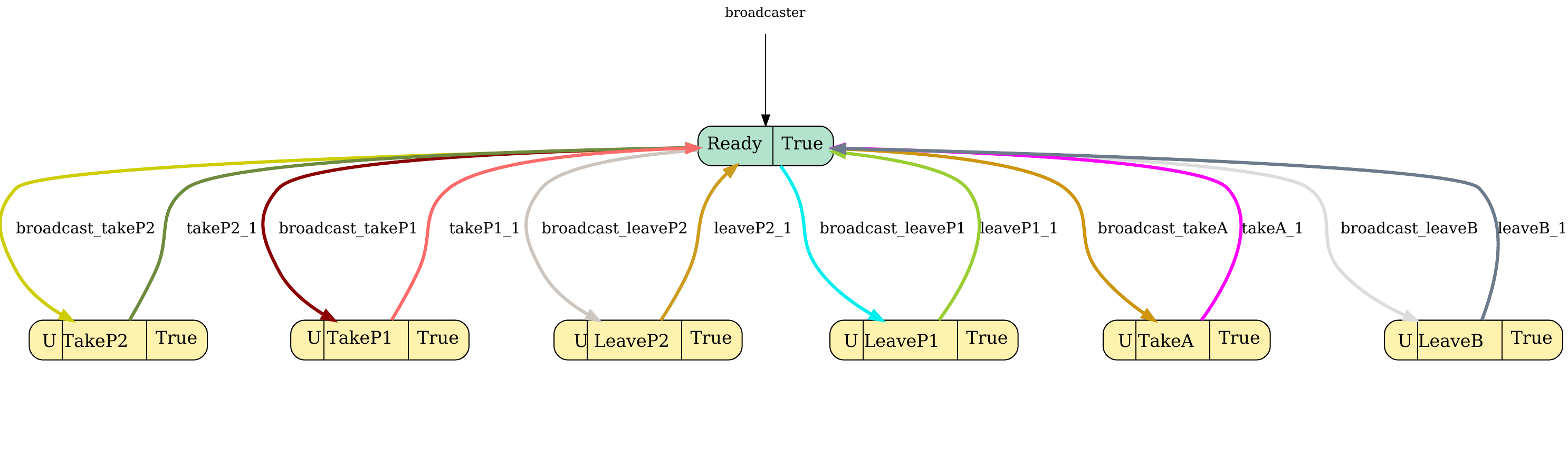}

  \caption{\imitator model for the Production Cell -- broadcaster component\\ (auxiliary automaton to facilitate broadcast communication)}
  \label{fig:model_broadcaster}
\end{figure}

\end{document}